\documentclass[12pt]{l4dc2021} 

\usepackage{etoc}

\usepackage{mathtools}
\usepackage{cleveref}
\usepackage{textcomp}
\usepackage{booktabs}

\usepackage{parskip}

\usepackage{algorithm}
\usepackage{algorithmic}
\usepackage{etoolbox}

\newtoggle{l4dc}
\togglefalse{l4dc}  

\newcommand{\C}{\mathbb{C}}
\newcommand{\K}{\mathbb{K}}

\newcommand{\R}{\mathbb{R}}
\newcommand{\E}{\mathbb{E}}

\newcommand{\pN}{\mathcal{N}}

\newcommand{\pr}{\mathbb{P}}
\newcommand{\norm}[1]{\left\lVert#1\right\rVert}
\newcommand{\normno}[1]{\lVert#1\rVert}
\newcommand{\abs}[1]{\left\lvert#1\right\rvert}

\newcommand{\data}{\mathcal{D}}

\newcommand{\vspacefigurecaptionhigh}{\vspace{-10pt}}
\newcommand{\vspacefigurecaptionlow}{\vspace{-6pt}}

\DeclareMathOperator{\Tr}{Tr}
\DeclareMathOperator{\diag}{diag}
\DeclareMathOperator*{\argmin}{argmin}

\DeclareMathOperator*{\var}{var}

\DeclareMathOperator*{\vectorized}{vec}

\title[Learning Stabilizing Controllers for LQR]{Learning Stabilizing Controllers for Unstable Linear Quadratic Regulators from a Single Trajectory}
\usepackage{times}

\author{\Name{Lenart Treven} \Email{trevenl@ethz.ch}\\
 \Name{Sebastian Curi} \Email{scuri@inf.ethz.ch}\\
 \Name{Mojmir Mutny} \Email{mmutny@inf.ethz.ch}\\
 \Name{Andreas Krause} \Email{krausea@ethz.ch}\\
 \addr ETH Zürich}

\begin{document}

\etocdepthtag.toc{mtchapter}
\etocsettagdepth{mtchapter}{subsection}
\etocsettagdepth{mtappendix}{none}

\maketitle

\begin{abstract}%
    The principal task to control dynamical systems is to ensure their stability. When the system is unknown, robust approaches are promising since they aim to stabilize a large set of plausible systems simultaneously. We study linear controllers under quadratic costs model also known as linear quadratic regulators (LQR). We present two different semi-definite programs (SDP) which results in a controller that stabilizes all systems within an ellipsoid uncertainty set. We further show that the feasibility conditions of the proposed SDPs are \emph{equivalent}. Using the derived robust controller syntheses, we propose an efficient data dependent algorithm -- \textsc{eXploration} -- that with high probability quickly identifies a stabilizing controller. Our approach can be used to initialize existing algorithms that require a stabilizing controller as an input while adding constant to the regret. We further propose different heuristics which empirically reduce the number of steps taken by \textsc{eXploration} and reduce the suffered cost while searching for a stabilizing controller. 
\end{abstract}


\section{Introduction}
\label{section:Introduction}

{\em Dynamical systems} are ubiquitous in real world applications, ranging from autonomous robots \citep{ExampleRobotics}, energy systems \citep{EnergyExample} to manufacturing \citep{ManufacturingExample}. Control theory \citep{trentelman2001control} seeks to find an optimal input to the system to ensure a desired behavior while suffering low cost. In particular, {\em linear} dynamical systems with quadratic costs can model a variety of practical problems \citep{Tornambe1998theory}, and enjoy an elegant solution referred to as {\em Linear Quadratic Regulator (LQR)}, whose history goes back to \citet{Kalman1960}. 

\looseness -1 
Despite the long and rich history of the LQR problem, {\em learning} dynamical systems and finding a stabilizing or optimal controller is still an actively studied problem. 
On one hand, there are systems that can be reset to an initial condition. 
For such systems, the multiple-trajectory (episodic) setting is natural and the exploration costs in unstable systems can be controlled by resetting the system. 
This setting is well studied and efficient algorithms rely on {\em certainty equivalent control} (CEC) \citep{Mania2019CertaintyEI}.
On the other hand, {\em unstable} systems that cannot be reset must be stabilized online from a {\em single} trajectory. 
After stabilization, there are different efficient algorithms that find an optimal controller \citep{simchowitz2020naive,Cohen2019SPDRelaxation,Abeille2020}.
Crucially, the algorithms that find an optimal controller require an initial stabilizing controller.
This privileged information is essential to ensure that unstable systems do not  ``explode''.
However, such prior knowledge is not always available. 

In this work, we address the problem of finding a stabilizing controller for a linear dynamical system in a single online trajectory. 
On the left plot of \Cref{figure: motivation}, we show the difference in cost between finding a stabilizing controller and an optimal one. 
When the true system is unstable, and no knowledge of a stabilizing controller is available, the system costs grow exponentially fast. 


\vspace{-1mm}

\paragraph{Contributions}
\looseness -1 
We extend the robust formulation of \citet{Dean2018OnSampleComplexity} from stabilizing all systems within some spectral norm around estimates, to the more general case when the synthesized controller stabilizes all systems within an ellipsoidal uncertainty set. We extend the obtained result to our main contribution where we prove equivalence between common robust controller synthesis algorithms. In particular we show that the controller synthesis of \citet{Umenberer2019}, which also tries to stabilize systems inside ellipsoid around estimates, and the derived SLS synthesis for ellipsoids share the same feasibility region: when one algorithm finds a stabilizing controller, so does the other one. 
Using the proposed robust controller syntheses applied to the ellipsoidal regions obtained from the Bayesian setting we propose an algorithm \textsc{eXploration}. The vanilla version synthesizes a stabilizing controller for the true underlying system in finite time with high probability. The vanilla \textsc{eXploration} approach probes the unknown system with zero-mean Gaussian actions. Additionally, we empirically show that with the robustly motivated choices of system probing which are different from zero-mean Gaussian, we reduce the length of the \textsc{eXploration} and {\em{substantially}} lower the total suffered cost. We demonstrate the practicality of our method on the standard common benchmark problems.

\begin{figure}
    \centering
    \includegraphics[width=0.8\textwidth]{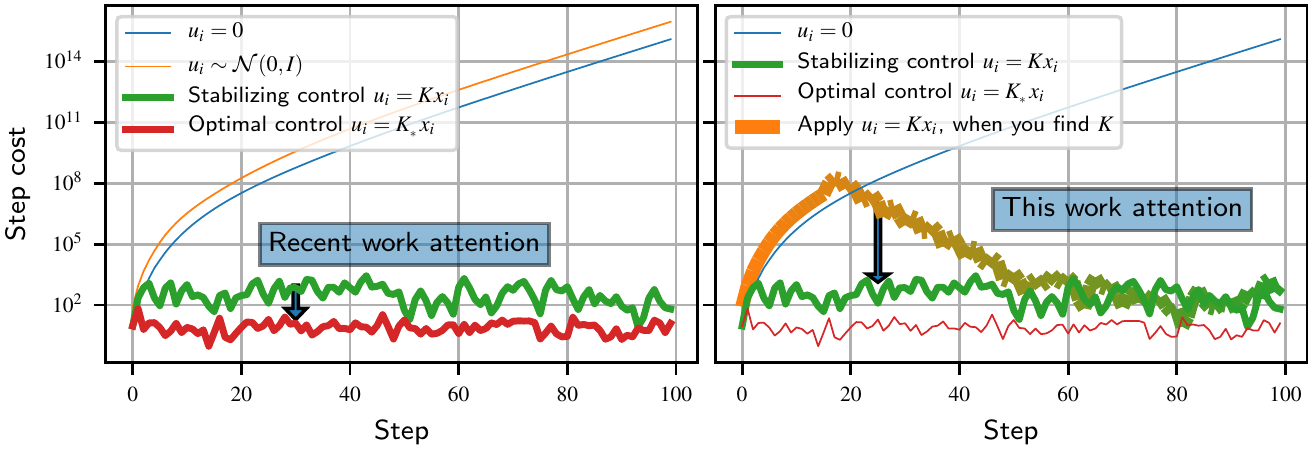}
    \vspacefigurecaptionhigh
    \caption{While recent work attention was mostly focused on how to adaptively progress from a given stabilizing controller to an optimal one, this work attention is how to find a stabilizing controller in the single-trajectory setting.\vspacefigurecaptionlow}
    \label{figure: motivation}
\end{figure}


\subsection{Related Work}
\vspace{-1mm}
Linear dynamical systems have been extensively studied in control theory (cf., \citet{Doyle1996}), nevertheless the interest reemerged in the Machine Learning community after the seminal work of \citet{abbasi2011}. \citet{Cohen2019SPDRelaxation,Mania2019CertaintyEI,simchowitz2020naive,Abeille2020,lale2020explore,faradonbeh2020optimism} and many others show that approaches relying on the optimism in the face of uncertainty (OFU) principle or certainty equivalence principle achieve $\mathcal{O}(\sqrt{T})$ regret in the online single trajectory setting with stochastic disturbances. However, most of the algorithms in this setting require the knowledge of an {\em initial stabilizing controller}, which might not always be available.
\looseness -1

\paragraph{System Identification}
The first step in \textsc{eXploration} is to {\em learn} the true system matrices $A_*, B_*$ using System Identification tools.  
\citet{simchowitz2018learning} show that the ordinary least squares (OLS) estimator attains a near optimal error rate $\mathcal{O}(1/\sqrt{i})$, where $i$ is the number of steps in the single trajectory setting for the case when $\rho(A_*) \le 1$\footnote{Spectral radius of matrix $A$ is defined as $\rho(A) = \max\{\abs{\lambda}|\exists v: Av = \lambda v\}$}. 
They further argue that more unstable systems are easier to estimate and prove exponential error decay for one dimensional unstable systems. 
\citet{FaradonbehIdentification} tackle more challenging general systems with eigenvalues everywhere but on the unit circle. 
They show that in this case for \emph{regular}\footnote{System is regular if every eigenvalue $\mu$ of $A_*$ with $\abs{\mu} >1$, has geometric multiplicity equal to 1.} systems the ordinary least squares (OLS) estimator is consistent. 
Further, \citet{Sarkar2018Identification} extend the OLS consistency to general regular systems. They show that the estimation error scales as $\mathcal{O}(1/\sqrt{i})$. 
\citet{Umenberer2019} show that the OLS is the same as the maximum (Gaussian) likelihood estimator. 
They further introduce an ellipsoid region around the estimates where the system lies with high confidence. 
We extend their idea to the Bayesian setting, where we assume a Gaussian prior on the system parameters, and show that, in this case, the maximum a-posteriori estimator is equivalent to the regularized least squares (RLS) estimator. 
Applying the analysis of \citet{Sarkar2018Identification} we show that the associated data dependent high probability credibility regions are consistent.
\looseness -1 

\paragraph{Controller Synthesis} 
The second step in \textsc{eXploration} is to synthesize a stabilizing controller for all systems in the credibility region that the System Identification step outputs. 
\citet{Dean2018OnSampleComplexity} derive a robust semi-definite program based on system level synthesis (SLS) whose  solution results in a stabilizing controller.
They use a multi-trajectory setting to build 2-ball confidence regions around the estimates.
We extend their algorithm to a tighter ellipsoidal region around the estimates.
It turns out that the SLS synthesis with an ellipsoidal region finds a stabilizing controller as the robust LQR synthesis proposed by \citet{Umenberer2019}. 
\citet{FaradonbehStabilization} propose a non-robust strategy and rely on well-known stability bounds of LQRs \citep{safonov1977gain}. 
The main contribution is that they identify the system in closed-loop by sampling different controllers from a Gaussian distribution so that they avoid irregularity of the closed loop matrix a.s. 
The main practical limitation is that one needs to specify the running time of the algorithm a-priori using unknown system parameters. 
On the other hand \textsc{eXploration} provably terminates in finite time solving a convex SDP without specifying an a-priori termination time.

\section{Problem Statement and Background}
\vspace{-1mm}
\label{section:Problem Statement}
We consider a system evolving with the following linear dynamics
\begin{align}
\label{equation:SystemEvolution}
    x_{i+1} = A_*x_i + B_*u_{i} + w_{i+1}, \quad x_0 = 0,
\end{align}
where $x_i \in \R^{d_x}$ are states, $u_i \in \R^{d_u}$ actions and $(w_i)_{i\ge 1} \stackrel{\mathclap{i.i.d.}}{\sim} \pN(0, \sigma_w^2I)$ unobserved Gaussian noise in $\R^{d_x}$. The matrices $A_*\in \R^{d_x\times d_x}, B_*\in \R^{d_x \times d_u}$ are unknown transition matrices. We sample actions $u_i$ from a policy $\pi$, which at every step $i$ maps the current history $((x_j)_{j\le i}, (u_j)_{j<i})$ to a distribution over actions. We assume that the system is \emph{stabilizable}, which means that there exists a matrix $K \in \R^{d_u \times d_x}$ such that $\rho(A_*+B_*K)<1$. At step $i$, we incur a cost $c_i$ given by
\begin{align}
\label{equation:StateCost}
    c_i = x_i^\top Q x_i + u_i^\top Ru_i,
\end{align}
where $Q \in \R^{d_x \times d_x}, R \in \R^{d_u \times d_u}$ are known positive definite matrices. 

When the system matrices $A_*, B_*$ are known, the optimal solution in the infinite horizon setting is given by the fixed map $u_i = K_*x_i$ and the optimal cost is $J_*$ \citep{bertsekas2000dynamic}. Hereby, $K_* = -(R + B_*^\top PB_*)^{-1}B_*^\top PA_*$, where $P$ is the solution to the {\em discrete algebraic Ricatti equation} of the system,  $P = DARE(A_*, B_*, Q, R)$.

While most of the recent work focuses on finding the optimal controller $K_*$ suffering the least possible cummulative cost, they require the knowledge of an initial {\em{stabilizing}} controller $K_0$. 
In this work, we focus on \textbf{finding a stabilizing controller in the single-trajectory setting}. 
If the system is unstable, the difference between using a stabilizing controller or not using a stabilizing controller results in an {\em exponential} difference in the suffered cost due to blow-up of the system. 
On the other hand, the difference between using a stabilizing controller and the optimal one results in a {\em linear} difference in the suffered cost as summarized in \Cref{figure: motivation}. 
Consequently, it is very desirable that a stabilizing controller is found quickly.

\section{Identifying A Stabilizing Controller}
\vspace{-1mm}
\label{section:Robust syntheses}
As we are trying to stabilize the system without knowing anything non-trivial about the dynamical system certain blow-up of the state from zero is inevitable. As this blow-up increases exponentially fast the state magnitude, it is essential that this period is kept very short. There are conceptually two variables we can influence as algorithm designers, a) what control signal we input to the system what we refer to as probing and b) stopping rule, which determines when we should stop probing and sufficient information about the system has been gathered to construct a stabilizing controller. 

In this work we focus on the latter and derive a {\em data-dependent stopping rule} based on feasibility of a semi-definite program. Our method is versatile and can be combined with any control inputs and we show that it terminates in finite time under zero-mean Gaussian control inputs. 
Our formalism is derived under a general assumption that we can construct estimates of the system matrices and ellipsoidal sets which with high probability contain the true system matrices or they serve as a surrogate for this task as is the case with the Bayesian approach.
More specifically,  we derive our results under the assumption that after playing a policy $\pi$ for $i$ steps we have estimates $(\widehat{A}_i, \widehat{B}_i)$ of the system $(A_*, B_*)$ and ellipsoid 
\begin{equation}\label{eq:ellipsoid}
\Theta_i = \{(A, B)| \Delta^\top D_i \Delta \preceq I, \Delta^\top = (A, B) - (\widehat{A}_i, \widehat{B}_i)\}
\end{equation}
around estimates for which we believe that $(A_*, B_*) \in \Theta_i$. 
Here $D_i$ is a data-dependent positive definite matrix. We present an example of how to construct such ellipsoidal region in \Cref{subsection: Bayesian credibility Region}.

In the following two subsections, we derive two different robust synthesis algorithms for uncertainty sets in the ellipsoidal region \eqref{eq:ellipsoid}. 
In \Cref{subsection:Equivalence}, prove that these two seemingly different approaches are actually equivalent in terms of robust stability.

\subsection{Robust System Level Synthesis (SLS)}
\label{subsection: Robust formulation from System Level Synthesis}
\vspace{-1mm}
\looseness -1 Our first stopping rule is based on a relaxation stemming from the SLS framework \citep{wang2019system}. In particular, we extend the work of \citet{Dean2018OnSampleComplexity} to ellipsoidal regions \eqref{eq:ellipsoid}. \citet{Dean2018OnSampleComplexity} show that a controller $K$ stabilizes {\em all} systems $(A, B) \in \Theta_i$ if for every $(A~B) \in \Theta_i$ we have:
\begin{align}
\label{equation:HInfinityConditionToStabilize}
    \norm{
    \Delta^\top
    \begin{pmatrix}
     I \\
     K
    \end{pmatrix}
    \left(
    zI-\widehat{A}_i-\widehat{B}_iK
    \right)^{-1}
    }_{\mathcal{H}_\infty} < 1,
\end{align}
where $\Delta^\top = (A~B)-(\widehat{A}_i~\widehat{B}_i)$.
The $\mathcal{H}_\infty$-norm for a function $f: \C \to \C^{d \times d}$  is defined as $\norm{f}_{\mathcal{H}_\infty} = \sup_{z \in \partial \mathbb{D}}\norm{f(z)}$, where $\mathbb{D} = \{z\in \C|\norm{z}<1\}$ is a unit disk in the complex plane.
With the current formulation we need to ensure that $\mathcal{H}_\infty$-norm constraint \eqref{equation:HInfinityConditionToStabilize} holds for every $(A, B) \in \Theta_i$. 
The main difference with \citet{Dean2018OnSampleComplexity} is that we apply the S-Lemma of \citet{LMIs2004} to obtain an equivalent formulation with a single $\mathcal{H}_\infty$-norm constraint using the ellipsoidal region instead of a 2-ball. Next, we transform the $\mathcal{H}_\infty$-norm constraint to a convex semi-definite constraint applying the KYP-Lemma \citep{KYPMoreReliableSource}. The equivalent feasibility problem reads:
\begin{equation}
\label{semi-definite constraint}
    \begin{aligned}
    \min_{X\succ 0, S, t\in (0, 1)} 0,  \quad
    &\qquad\text{s.t.} \quad 
     \begin{pmatrix}
    X - I & \widehat{A}_iX+\widehat{B}_iS & 0  \\
    (\widehat{A}_iX+\widehat{B}_iS)^\top & X & \begin{pmatrix} X \\ S \end{pmatrix}^\top\\
    0 & \begin{pmatrix} X \\ S \end{pmatrix} &tD \\
    \end{pmatrix}
    \succeq 0.
    \end{aligned}
\end{equation}
The stabilizing controller is extracted from the solution of the Robust SLS \eqref{semi-definite constraint} as $K=S X^{-1}$.  We show the derivation details in the \Cref{appendix: Semi Definite Program from SLS}\iftoggle{l4dc}{ of the extended paper \citep{treven2020learning}}{}.

\subsection{Robust Linear Quadratic Regulator (LQR)}
\label{subsection: Robust formulation from Semi-Definite Program}
\vspace{-1mm}
The derivation of our second robust controller synthesis is based on the reformulation of the LQR problem, which finds the optimal infinite horizon controller. This reformulations lends itself to an efficient SDP relaxation \citep{boyd1994linear}. 
We follow the exposition from \citet{Cohen2018}, assuming we know matrices $A_*, B_*$ we can obtain the optimal infinite controller $K_*$ by first solving
\begin{equation}
\label{optimal infinite hoizon policy SDP ineq}
\begin{aligned}
\min_{\Sigma \succeq 0}
    &\Tr\left(
    \begin{pmatrix}
        Q & 0 \\
        0 & R
    \end{pmatrix}
    \Sigma
    \right)\\
    &\text{s.t.} \quad
    \Sigma_{xx} \succeq (A_* ~ B_*)\Sigma (A_* ~ B_*)^\top  + \sigma_w^2 I,
\end{aligned}
\end{equation}
and then extracting the optimal controller as $K_* = \Sigma_{ux}\Sigma_{xx}^{-1}$. Here $\Sigma = \begin{pmatrix}\Sigma_{xx} & \Sigma_{xu} \\ \Sigma_{ux} & \Sigma_{uu}\end{pmatrix}$, where $\Sigma_{xx} \in \R^{d_x \times d_x}$ and $\Sigma_{uu} \in \R^{d_u \times d_u}$, represents the joint covariance matrix of the state and action. The derivation with the motivation behind the SDP \eqref{optimal infinite hoizon policy SDP ineq} is given in \Cref{appendix: Optimal Infinite Horizon via SDP}\iftoggle{l4dc}{ of the extended paper \citep{treven2020learning}}{}, where we also show in \Cref{lemma: Better covariance matrix for Epsilon region} that the semi-definite constraint in the SDP \eqref{optimal infinite hoizon policy SDP ineq} ensures that controller synthesized as $K = \Sigma_{ux}\Sigma_{xx}^{-1}$ stabilize the system $A_*, B_*$. Inspired by \citet{Umenberer2019}, the robust formulation of the SDP problem \eqref{optimal infinite hoizon policy SDP ineq} is:
\begin{equation}
\label{SDP: stabilizing SDP}
    \begin{aligned}
        \min_{\Sigma \succeq 0} &\Tr\left(
    \begin{pmatrix}
    Q & 0 \\
    0 & R
    \end{pmatrix}
    \Sigma
    \right)\\
    &\text{s.t. } \forall (A, B) \in \Theta_i: ~ 
    \Sigma_{xx} \succeq (A~B)\Sigma(A~B)^\top + \sigma_w^2 I.
    \end{aligned}
\end{equation}
As in \Cref{subsection: Robust formulation from System Level Synthesis} we have to ensure that one condition has to hold for every system in the ellipsoid $\Theta_i$. Applying the S-Lemma we reformulate problem given by \cref{SDP: stabilizing SDP} to an equivalent convex SDP:
\begin{equation}
\label{SDP: stabilizing SDP final}
    \begin{aligned}
        \min_{\Sigma \succeq 0, t \ge 0} &\Tr\left(
    \begin{pmatrix}
    Q & 0 \\
    0 & R
    \end{pmatrix}
    \Sigma
    \right)\\
    &\text{s.t. } 
    \begin{pmatrix}
        \Sigma_{xx} - (\widehat{A}_i~\widehat{B}_i)\Sigma (\widehat{A}_i~\widehat{B}_i)^\top - (t + \sigma_w^2) I & (\widehat{A}_i~\widehat{B}_i)\Sigma \\
        \Sigma (\widehat{A}_i~\widehat{B}_i)^\top & t D - \Sigma
    \end{pmatrix}\succeq 0.
    \end{aligned}
\end{equation}
The stabilizing controller is extracted from the optimal solution as $K = \Sigma_{ux}\Sigma_{xx}^{-1}$.

\subsection{Equivalence} \label{subsection:Equivalence}
At first glance one could think that we have derived two completely different control synthesis procedures, however as we will see this is not the case and the two optimization problems have the same feasible regions.

\begin{theorem}
\label{theorem: feasibility equivalence}
    The Robust SLS given by \Cref{semi-definite constraint} has a nonempty solution if and only if the Robust LQR given by \Cref{SDP: stabilizing SDP final} has a nonempty solution.
\end{theorem}

The proof of the theorem is provided in \Cref{appendix: Proof of Theorem 1}\iftoggle{l4dc}{ of the extended paper \citep{treven2020learning}}{}. Despite the fact that the two SDPs \eqref{semi-definite constraint} and \eqref{SDP: stabilizing SDP final} have the same feasible solutions, they differ in the objective. The robust SLS \eqref{semi-definite constraint} is maximizing a constant for which we can find a stabilizing controller, whereas the Robust LQR \eqref{SDP: stabilizing SDP final} is minimizing the upper bound of the maximal infinite horizon cost of the systems in $\Theta_i$. The specific nature of the objective is not interesting for the stopping rule, however can have practically dramatic impact in downstream tasks. For different possible objective for the SLS synthesis \eqref{semi-definite constraint} please refer to \Cref{section: Controlled action selection in Phase 1}.

\subsection{Example: Bayesian credible sets} 
\label{subsection: Bayesian credibility Region}
 In this section, we show a particular design choice of estimators $\widehat{A}_i, \widehat{B}_i$ and region $\Theta_i$ which results from the Bayesian setting. Inspired by the work of \citet{Umenberer2019}, we place a Gaussian prior\footnote{Regarding the sense of this assumption and cases when this assumption fails look at \Cref{section: Data Dependent credibility Region}} on the system matrices $A_*, B_*$ i.e. $\vectorized(A_*, B_*) \sim \pN\left(0, \sigma_w^2/\lambda I\right)$. In this case we have explicit formulas for the posterior distribution of $\vectorized(A_*, B_*)|(x_j)_{j\le i}, (u_j)_{j<i}$. As derived in the \Cref{appendix: Posterior Distribution}\iftoggle{l4dc}{ of the extended paper \citep{treven2020learning}}{} it turns out the posterior is also Gaussian and the MAP estimator $\vectorized(\widehat{A}_i, \widehat{B}_i)$ is exactly the RLS estimator:
 \begin{align}
 \label{MAP estimator}
    \widehat{A}_i, \widehat{B}_i = \argmin_{A, B}\sum_{j=0}^{i-1}\norm{x_{j+1} - Ax_{j}-Bu_j}_2^2 + \lambda \norm{(A~B)}_F^2.
 \end{align}
 Moreover, the Bayesian credibility region is $\Theta_i= \{(A~B)|\Delta^\top D_i \Delta \preceq I, \Delta^\top = (A~B) - (\widehat{A}_i, \widehat{B}_i)\}$. Here $D_i$ represents the scaled inverse covariance matrix of the posterior distribution and is explicitly given as $D_i = \frac{1}{c_\delta\sigma_w^2}\left(\sum_{j=1}^{i}z_jz_j^\top + \lambda I\right)$, where $z_j^\top = (x_j^\top u_j^\top)$ and $c_\delta$ is the $(1-\delta)$-quantile of the $\chi^2$ distribution with $d_x(d_x + d_u)$ degrees of freedom. With this definition of $\Theta_i$ we have $(A_*~B_*) \in \Theta_i$ w.p. $1-\delta$.
\section{\textsc{eXploration} Algorithm}
\label{section:ALgorithm}
\looseness -1 

We now show how to use the derived results (c.f., \Cref{section:Robust syntheses}) to {\em provably} find a robust controller in the Bayesian setting. In the \Cref{appendix: Initialization of Existing Algorithms}\iftoggle{l4dc}{ of the extended paper \citep{treven2020learning}}{} we show how to initialize the algorithms, specifically OSLO \citep{Cohen2019SPDRelaxation} and CEC \citep{simchowitz2020naive}, which need a stabilizing controller as an input, with the proposed \textsc{eXploration} algorithm.

\begin{algorithm}[H]
    \caption{\textsc{eXploration}}
    \label{algorithm:FirstPart}
    \begin{algorithmic}[]
        \STATE {\textbf{Input:}}{ $x_0 =0, \lambda, \delta$}
        \FOR{$i=1, \ldots $}{
            \STATE {\color{blue} /* Probing Signal /*}
            {\STATE Play $u_{i} \sim \pi(\cdot | x_{1:i},u_{1:i-1})$ and observe state $x_{i+1}$}.
            \STATE {\color{blue} /* Stopping Rule /*}
            {\STATE Build a confidence region $\Theta_i$, such that $(A_*~B_*) \in \Theta_i$ w.p. $1-\delta$. (c.f. \Cref{subsection: Bayesian credibility Region}) 
            }
            {\STATE Solve a robust controller synthesis $\forall (A, B) \in \Theta_i$. (c.f. \Cref{subsection: Robust formulation from System Level Synthesis} or \Cref{subsection: Robust formulation from Semi-Definite Program}) }
            \STATE{\textbf{if} a controller is found \textbf{return} stabilizing controller $K_0$}
            }
        \ENDFOR
    \end{algorithmic}
\end{algorithm}
\vspace{-6mm}
The learner explores using a policy $\pi$ that only depends on the past states and inputs. 
Using the collected data, it builds an empirical estimate and a confidence region around it. 
Finally, it attempts to solve a robust controller synthesis problem. 
If it fails, the algorithm continues. 
If it succeeds, the algorithm terminates and returns a provably stabilizing controller for the {\em true} underlying system.

The credibility regions must contain the true system with probability $1-\delta$ only at the time $i^*$ in which a controller is found and not uniformly over all time steps. 
This is crucial as it allows us to use tight credibility regions. 
Then, with probability $\delta$, the algorithm might fail to return a stabilizing controller, and it will return a stabilizing controller with probability $1-\delta$.

\looseness -1 In this section, we analyze the \textsc{Vanilla eXploration} variant, 
in which the policy is to choose independent zero-mean Gaussian action, i.e., $u_i \sim \pi(\cdot | x_{1:i},u_{1:i-1}) = \mathcal{N}(0, \sigma_u^2I)$. 
Next, we prove that \textsc{Vanilla eXploration} finishes in $\widetilde{O}(1)$ time.
In \Cref{section: Controlled action selection in Phase 1}, we discuss different probing heuristics that perform well in practice but where we lose the finte time termination guarantee. 
Nevertheless, the algorithm still remains valid: if it terminates, the resulting controller provably stabilizes the system.

\begin{theorem}
\label{theorem: Algorithm termination}
Assuming the aforementioned setting, then with probability $1-\delta$ \textsc{Vanilla eXploration} returns a stabilizing controller for $(A_*, B_*)$ in time:
\begin{align}
    \widetilde{\mathcal{O}}\left(\operatorname{polylog}(\delta)(1 + \|K\|_2)^2\norm{(zI - A_*-B_*K)^{-1}}_{\mathcal{H}_\infty}^2\right) \label{eq:Finite time Guarantee},
\end{align}
where $K$ is any stabilizing static controller. 
\end{theorem}

\paragraph{Proof Sketch:}
\looseness -1 First, note that since we assume that entries of $A_*, B_*$ are sampled from independent Gaussian, system $(A_*, B_*)$ is stabilizable a.s. If $K$ is a stabilizing controller we derive in the \Cref{appendix: Proof of Theorem 2}\iftoggle{l4dc}{ of the extended paper \citep{treven2020learning}}{}, extending the results of \citet{Dean2018OnSampleComplexity}, that the SLS synthesis \eqref{semi-definite constraint} is feasible if 
\begin{align}
\label{lower bound for D}
    \mathcal{O}\left((1 + \norm{K}_2)^2\norm{(zI - A_*-B_*K)^{-1}}_{\mathcal{H}_\infty}^2\right) \le \lambda_{min}(D).  
\end{align}
\looseness -1To obtain the best bound we choose a stabilizing controller $K$ such that the left hand side of \Cref{lower bound for D} is minimized. From \Cref{lower bound for D} follows that as soon as the smallest eigenvalue of the matrix $D$ is large enough, the robust synthesis will be feasible. At the same time from the analysis of \citet{Sarkar2018Identification} follows that $\widetilde{\Omega}(i)I \preceq D_i$ for every regular system. Again, since we assume a Gaussian prior on $(A_*, B_*)$, the system is regular a.s. Assembling the pieces together we arrive at the result, for which we provide more detailed proof in the \Cref{appendix: Proof of Theorem 2}\iftoggle{l4dc}{ of the extended paper \citep{treven2020learning}}{}.
\subsection{Different probing signals with \textsc{eXploration}}
\label{section: Controlled action selection in Phase 1}

\looseness -1 The \textsc{VANILLA eXploration} approach takes random actions $u_i \sim \pN(0, \sigma_u^2I)$. For such a choice we can guarantee that \Cref{algorithm:FirstPart} terminates after constant time, depending only on the system parameters. 
However, as we demonstrate in our experiments (c.f., \Cref{section: Additional Experiments} of the extended paper \citep{treven2020learning}), the states grow {\em exponentially} during this phase, which can be highly problematic for certain applications. 
We now propose improved, {\em data-dependent} policies to counteract this blow-up. In particular, we consider playing $u_i \sim \pN(K_ix_i, \sigma_u^2I)$, where $K_i$ is a controller picked at time $i$. With such a controller, we generally lose the theoretical guarantee that the \Cref{algorithm:FirstPart} will terminate. However, the data dependent credibility region on estimation errors from \cref{subsection: Bayesian credibility Region} (and thus the validity of the stopping condition) is still valid and we can run \Cref{algorithm:FirstPart}. 
With data dependent inputs, we cannot guarantee that the minimum eigenvalue of $D_i$ grows as $\tilde{\Omega}(i)$. 
Next, we discuss different choices for controller $K_i$ that we study in our experiments.

\paragraph{\textsc{CEC}} 
As first possibility, we could act as if the estimators $\widehat{A}_i, \widehat{B}_i$ are the true system matrices and we compute the controller $K_i$ as the optimal controller: 
\begin{equation}
    K_i = -(R + \widehat{B}_i^\top P\widehat{B}_i)^{-1}\widehat{B}_i^\top P\widehat{A}_i, \label{eq:CEC}
\end{equation}
where $P_i = \operatorname{DARE}(\widehat{A}_i, \widehat{B}_i, Q, R)$, i.e., we act using  Certainty Equivalent Control (\textsc{CEC}).

\paragraph{\textsc{MinMax}} 
For the second $K_i$ we consider controller which minimizes the maximal closed loop norm of the systems in $\Theta_i$. At every time step we synthesize the controller $K_i$ as
\begin{align}
\label{minmax controller definition}
    K_i = \argmin_K\max_{(A, B) \in \Theta_i} \norm{A+BK}_2
\end{align}
The controller defined in \Cref{minmax controller definition} can be efficiently computed via a convex SDP. We derive the convex SDP formulation of the min max problem given by \Cref{minmax controller definition} in \Cref{appendix: Convex SDP formulation for MinMax} of the extended paper \citep{treven2020learning}. 

\paragraph{\textsc{RelaxedSLS}} 
As a third alternative we relax the constraint $t \in (0, 1)$ to $t \ge 0$ in the SDP feasibility problem \eqref{semi-definite constraint}, and minimize the value of $t$, i.e.:

\begin{equation}
\label{relaxedSLS}
    \begin{aligned}
    &\min_{X\succ 0, S, t \ge 0}t
    &\quad\text{s.t. semi-definite constraint \eqref{semi-definite constraint}} 
    \end{aligned}
\end{equation}
The controller is then synthesized as $K_i = SX^{-1}$.
With such relaxation, the SDP is always feasible. 
The interpretation of this relaxation is that when $t \ge 1$ we find a controller that stabilizes all systems $(A, B)$ in a {\em smaller} confidence region around the estimates $(\widehat{A}_i, \widehat{B}_i)$.
Furthermore, this algorithms returns a provably stabilizing controller  when $t < 1$. 
Although in principle we could also increase $D_i$ in the Robust LQR synthesis in \eqref{SDP: stabilizing SDP final}, this requires a tedious exponential line search, whereas the \textsc{RelaxedSLS} synthesis does this automatically. 

\section{Experiments}
\label{section: Experiments}

\looseness -1 In this section, we critically evaluate the different components of \textsc{eXploration} empirically. 
In \Cref{section: Data Dependent credibility Region}, we investigate when the credibility regions are correct and when do they fail on a fixed system $A_*, B_*$. 
In particular, the algorithm fails when the prior parameter $\lambda$ is too large. 
To overcome this issue, we suggest a way of selecting the prior parameter $\lambda$, given some mild privileged information. 
In \Cref{section: eXploration Performance}, we compare the time it takes to find a stabilizing controller and the total cost suffered using different probing signals. 
Although \textsc{Vanilla eXploration} provably terminates, the heuristic variants perform better in practice. 
In all the considered examples the cost matrices $Q$ and $R$ are equal to the identity matrix of the appropriate dimensions. 
The scales of unobserved and played noise covariance matrices are $\sigma_w^2=\sigma_u^2 = 1$. 
We set the probability of failure to $\delta = 0.1$.


\subsection{Data Dependent Credibility Region}
\label{section: Data Dependent credibility Region}
\looseness -1 To illustrate how \textsc{eXploration} builds the credibility regions, we consider a one dimensional system $A_* = 1.5, B_* = 1.8$. 
We select $\lambda = \frac{1}{4}$ and $\lambda = 3$ and show consecutive credibility regions in \Cref{figure: moving credibility regions}. 
As we can see, the credibility region $\Theta_i$ shrinks as we see more data. 
For both choices of $\lambda$, Robust SLS and Robust LQR become feasible after 4 iterations and the algorithm terminates. 
Crucially, when $\lambda$ is too large (i.e. too small variance of the prior), it may happen that the true parameters $A_*, B_*$ are not inside the credibility region as we can see on in the middle subfigure of \Cref{figure: moving credibility regions}. 

\begin{figure}[ht]
    \centering
    \includegraphics[width=\textwidth]{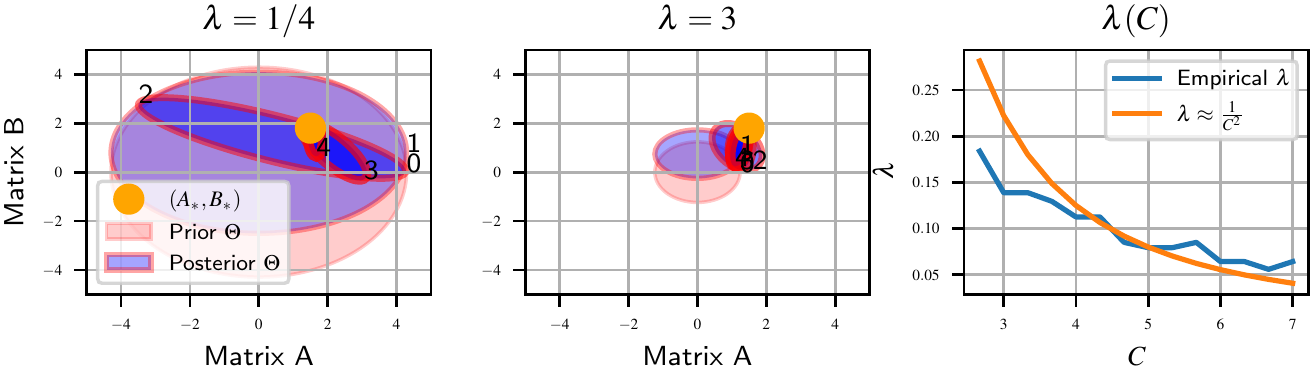}
    \vspacefigurecaptionhigh
    \vspace{-4mm}
    \caption{When credibility region $\Theta_i$ is small enough SDP \eqref{SDP: stabilizing SDP final} finds a controller which stabilizes every system in $\Theta_i$. If we choose a Gaussian prior with too small covariance matrix the credibility regions $\Theta_i$ might not contain the true system and the resulting controller will not stabilize it. If we know a bound $C$ on the Frobenious norm of the system, experiments show that a reasonable choice for $\lambda$ is $\lambda \approx \frac{1}{C^2}$.\vspacefigurecaptionlow}
    \label{figure: moving credibility regions}
\end{figure}

This means that selecting $\lambda$ is a problem-dependent quantity (as any prior). 
However, if we assume that we have of a constant $C$ such that $\norm{(A_*~ B_*)}_F \le C$,  then in the \Cref{section: Bounds failure} \iftoggle{l4dc}{ of the extended paper \citep{treven2020learning}}{} we suggest that selecting $\lambda \approx 1/C^2$, results in the credibility regions which empirically contain the true system with probability at least $1-\delta$.
On the right most subfigure of \Cref{figure: moving credibility regions} we plot the largest $\lambda$ for which we empirically observe that the one dimensional systems with $\norm{(A_*~B_*)} \le C$ are inside regions $\Theta_i$ with empirical probability at least $1-\delta$.

\subsection{\textsc{eXploration} Performance}
\label{section: eXploration Performance}

Next we will illustrate the cost suffered and time until we find a stabilizing controller on a system
\begin{align}
    \label{equation:SystemFromDean2017}
    A_* =
    \begin{pmatrix}
    1.01 & 0.01 & 0 \\
    0.01 & 1.01 & 0.01 \\
    0 & 0.01 & 1.01
    \end{pmatrix}, \quad
    B_* = I
\end{align}
introduced by \cite{Dean2018OnSampleComplexity}, here we use $\lambda=1$. We show a sample run with different heuristics in the \Cref{section: Additional Experiments}\iftoggle{l4dc}{ of the extended paper \citep{treven2020learning}}{}, here we show in \Cref{table: comparison of steps and cost on Dean system} more in-depth analysis of the \textsc{eXploration} performances on system~\eqref{equation:SystemFromDean2017}.

During exploration the agent suffers quadratic costs \Cref{equation:StateCost} and at the \textsc{eXploration} termination leaves the system in state $x_T$. We define the total cost of exploration as
\begin{equation}
    \text{Cost} = \sum_{i=1}^{T} \left(x_i^\top Q x_i + u_i^\top R u_i\right) + x_T^\top Px_T,
\end{equation}
where $P = \operatorname{DARE}(A_*, B_*, Q, R)$, and $T$ is the termination time of \textsc{eXploration}. Since the cost can grow exponential with time we will report the logarithm of the Cost. 

We compare robust synthesis with ellipsoidal bounds to two benchmarks. The first is robust synthesis with 2-ball estimation bounds of \citet{Dean2018OnSampleComplexity}. For the second benchmark we compare robust controller to the CEC which was used as a stabilizing controller in e.g. \citet{FaradonbehStabilization,simchowitz2020naive}. In particular we analyze how large region around estimates they stabilize. For both controllers we use the tightest ellipsoidal bounds. To compute the stopping time when CEC stabilizes all systems inside the ellipsoidal bound we sample 1000 systems from the ellipsoid boundary and if CEC stabilizes all we stop \textsc{eXploration}\footnote{Note that it can still happen that CEC does not stabilize all systems inside ellipsoid, however already with this approximation robust controller stabilizes larger ellipsoidal region.}. 

Using ellipsoidal compared to 2-ball bounds significantly reduces the number of steps of \textsc{eXploration}. Consequently we also suffer much less exploration cost. CEC naturally stabilizes some region around the estimates, however as we can see on the \Cref{table: comparison of steps and cost on Dean system} CEC stabilizes smaller region compared to the robust controller, which is paramount in the case when the cost grows exponentially.
\begin{table}[H]
\centering
\vspace{-2mm}
\caption{Using ellipsoidal confidence regions significantly shortens the exploration time and consequently the cost. Robust controller stabilizes larger region than CEC which is crucial when the cost grows exponentially. We report median $\pm$ standard deviation.}
\label{table: comparison of steps and cost on Dean system}
\begin{tabular}{@{}lllllll@{}}
                  \toprule
                  & \multicolumn{2}{c}{Ellipsoidal region} & \multicolumn{2}{c}{2-ball region} & \multicolumn{2}{c}{CEC as stopping time} \\ \midrule
       & Steps          &  $\log(\text{Cost})$ & Steps        & $\log(\text{Cost})$ & Steps         & $\log(\text{Cost})$ \\ \midrule
\textsc{Vanilla} & $\mathbf{44\pm8.2}$   & $\mathbf{8.7\pm0.58}$  & $84\pm18$   & $11\pm1.1$  & $50\pm12$     & $9.3\pm0.73$     \\
\textsc{CEC}     & $\mathbf{25\pm8.5}$ & $\mathbf{6.2\pm1.1}$  & $110\pm32$ & $7.4\pm0.62$ & $53\pm14$ & $6.8\pm0.56$ \\
\textsc{MinMax} & $\mathbf{24\pm9}$   & $\mathbf{7.1\pm3.2}$  & $120\pm39$ & $7.9\pm3.1$  & $58\pm22$ & $7.8\pm7.9$  \\
\textsc{RelaxedSLS}       & $\mathbf{26\pm8.9}$   & $\mathbf{7.2\pm3}$     & $170 \pm 86$  & $9.1\pm2.9$ & $78\pm36$     & $8.6\pm3.2$      \\ \bottomrule
\end{tabular}
\end{table}
\section{Discussion and Conclusions}
\label{section: Conclusions}
In Section \ref{section:Robust syntheses} we presented two seemingly different relaxation techniques for solving Riccati equation under uncertainty. The two relaxations one in Z-transform space and the other convex relaxation in the SDP formulation lead to the same feasible regions. We instantiated our stopping rule with uncertainty regions over the system matrices constructed via Bayesian means, however the stopping rule is more general and can be used with confidence estimates constructed without prior assumptions assumption once we can guarantee uncertainty sets otherwise. To the best of our knowledge, provable anytime adaptive consistent confidence estimates for unstable linear systems are not known, but should these be constructable our stopping rule can be used with them.

\small

\bibstyle{apalike}
\bibliography{refs.bib}

\normalsize

\newpage 
\appendix

\renewcommand*\contentsname{Contents of Appendix}
\etocdepthtag.toc{mtappendix}
\etocsettagdepth{mtchapter}{none}
\etocsettagdepth{mtappendix}{subsection}
\tableofcontents

\newpage

\section{Posterior Distribution}
\label{appendix: Posterior Distribution}

In this section we will show how to obtain regions around RLS estimates $\widehat{A}, \widehat{B}$ where $A_*, B_*$ lies with high probability. The regions will be of the form 
\begin{align}
\label{high probability region ellipsoid}
    \left\{(A, B) | \Delta^\top D \Delta \preceq I, \Delta^\top = (A~B) - (\widehat{A}~\widehat{B})\right\},
\end{align}
where $D$ is a positive definite matrix which depends on the observed past states and played actions. To derive the posterior distribution we were inspired by the work of \citet{Umenberer2019} that assumed the uninformative "improper" prior. Although majority of the steps are identical to their derivation we include it for the sake of completeness. 
The region given by the \cref{high probability region ellipsoid} represents an ellipsoid around $(\widehat{A}~\widehat{B})$. First we show lemma which converts the \cref{equation:SystemEvolution} to a form which we will use in the derivation of a posterior belief.

\begin{lemma}
 Denoting $\Phi_i = (x_i^\top~u_i^\top) \otimes I_{d_x}$ and $\vartheta_* = \vectorized(A_*, B_*)$ we can rewrite \cref{equation:SystemEvolution} as:
 \begin{align}
     x_{i+1} = \Phi_i\vartheta_* + w_{i+1}
 \end{align}
\end{lemma}

\begin{proof}
Compute:
    \begin{align*}
        x_{i+1} &= (A_* ~ B_*)
        \begin{pmatrix}
        x_i \\ 
        u_i
        \end{pmatrix}
        + w_{i+1} \\
        &= \vectorized\left(
        (A_* ~ B_*)
        \begin{pmatrix}
        x_i \\ 
        u_i
        \end{pmatrix}
        \right) + w_{i+1}\\
        &= \left((x_i^\top~u_i^\top) \otimes I_{d_x} \right) \vectorized((A_*~B_*)) + w_{i+1} \\
        &=  \Phi_i\vartheta_* + w_{i+1}
    \end{align*}
\end{proof}

\subsection{Exact posterior}
\label{paragraph: Computation of exact posterior}
To compute the posterior distribution we first observe:
\begin{align*}
    p(\vartheta_*|\data) \propto p(\data|\vartheta_*)p(\vartheta_*).
\end{align*}
We first compute $p(\data |\vartheta_*)$. By the product rule we have:
\begin{align*}
    p(\data|\vartheta_*) \propto \prod_{j=1}^ip(x_j|x_{j-1}, u_{j-1}, \vartheta_*).
\end{align*}
Since $p(x_j|x_{j-1}, u_{j-1}, \vartheta_*)$ is the density of $\pN(\Phi_{j-1}\vartheta_*, \sigma_w^2I)$ we further have:
\begin{align*}
    p(\data|\vartheta_*) &\propto \prod_{j=1}^i e^{-\frac{1}{2\sigma_w^2}(x_j - \Phi_{j-1}\vartheta_*)^\top(x_j - \Phi_{j-1}\vartheta_*)} \\
    &= \exp\left(-\frac{1}{2\sigma_w^2}\sum_{j=1}^i\norm{x_j - \Phi_{j-1}\vartheta_*}^2\right) \\
    &\propto \exp\left(-\vartheta_*^\top \left(\frac{1}{2\sigma_w^2}\sum_{j=1}^i\Phi_{j-1}^\top\Phi_{j-1}\right)\vartheta_* + \left(\frac{1}{\sigma_w^2}\sum_{j=1}^ix_j^\top \Phi_{j-1}\right)\vartheta_* \right).
\end{align*}
Together with prior
\begin{align*}
    p(\vartheta_*) \propto \exp\left(-\vartheta_*^\top \left(\frac{\lambda}{2\sigma_w^2}I\right)\vartheta_*\right)
\end{align*}
we obtain:
\begin{align*}
     p(\vartheta_*|\data) \propto \exp\left(-\vartheta_*^\top \left(\frac{1}{2\sigma_w^2}\sum_{j=1}^i\Phi_{j-1}^\top\Phi_{j-1}+ \frac{\lambda}{2\sigma_w^2}I\right)\vartheta_* + \left(\frac{1}{\sigma_w^2}\sum_{j=1}^ix_j^\top \Phi_{j-1}\right)\vartheta_* \right).
\end{align*}
Matching the coefficients we realize that $\vartheta_*|\data \sim \pN(\mu, \Sigma)$, where:
\begin{align*}
    \Sigma^{-1} &= \frac{1}{\sigma_w^2}\sum_{j=1}^i\Phi_{j-1}^\top\Phi_{j-1}+ \frac{\lambda}{\sigma_w^2}I \\
    \mu &= \left(\frac{1}{\sigma_w^2}\sum_{j=1}^i\Phi_{j-1}^\top\Phi_{j-1}+ \frac{\lambda}{\sigma_w^2}I\right)^{-1}\left(\frac{1}{\sigma_w^2}\sum_{j=1}^i \Phi_{j-1}^\top x_j\right) \\
    &= \left(\sum_{j=1}^i\Phi_{j-1}^\top\Phi_{j-1}+ \lambda I\right)^{-1}\left(\sum_{j=1}^i\Phi_{j-1}^\top x_j\right)
\end{align*}
For the estimators $\widehat{A}, \widehat{B}$ we than take MAP, for which we have $\vectorized((\widehat{A}~\widehat{B})) = \mu$. We will now derive explicit value of $(\widehat{A}~ \widehat{B})$. We denote by $z_j = (x_j^\top~u_j^\top)^\top$ and compute:
\begin{align*}
\vectorized((\widehat{A}~\widehat{B})) &= \left(\sum_{j=1}^i\Phi_{j-1}^\top\Phi_{j-1}+ \lambda I\right)^{-1}\left(\sum_{j=1}^i \Phi_{j-1}^\top x_j\right) \\
&= \left(\sum_{j=1}^i(z_{j-1}\otimes I_{d_x})(z_{j-1}^\top \otimes I_{d_x} )+ \lambda I\right)^{-1}\left(\sum_{j=1}^i(z_{j-1}\otimes I_{d_x})x_j\right) \\
&=  \left(\left(\sum_{j=1}^iz_{j-1}z_{j-1}^\top + \lambda I_{d_x+d_u}\right)\otimes I_{d_x} \right)^{-1}\vectorized\left(\sum_{j=1}^ix_jz_{j-1}^\top\right) \\
&=  \left(\left(\sum_{j=1}^iz_{j-1}z_{j-1}^\top + \lambda I_{d_x+d_u}\right)^{-1}\otimes I_{d_x}\right) \vectorized\left(\sum_{j=1}^ix_jz_{j-1}^\top\right) \\
&= \vectorized\left( \left(\sum_{j=1}^ix_jz_{j-1}^\top\right)\left(\sum_{j=1}^iz_{j-1}z_{j-1}^\top + \lambda I_{d_x+d_u}\right)^{-1}\right).
\end{align*}
Hence we obtained that MAP estimator satisfy: 
\begin{align*}
    (\widehat{A}~\widehat{B}) = \left(\sum_{j=1}^ix_jz_{j-1}^\top\right)\left(\sum_{j=1}^iz_{j-1}z_{j-1}^\top + \lambda I_{d_x+d_u}\right)^{-1}
\end{align*}
which is the same if we would compute the RLS estimator with regularizing parameter $\lambda$.

\subsection{High probability regions}
\label{paragraph: High probability error bounds}
 Using the exact posterior distribution computed in the \Cref{paragraph: Computation of exact posterior} we have that $(A_*, B_*) \in E = \{(A, B)| \theta = \vectorized((A~B)), (\theta - \mu)^\top\Sigma^{-1}(\theta - \mu) \le c_{\delta}\}$ w.p. at least $1-\delta$, where
$c_\delta$ is chosen in such a way that for $Z\sim \chi^2_{d_x^2 + d_xd_u}$ we have: $\mathbb{P}(Z\ge c_\delta) = \delta$. For matrices $A, B$ denote by $\Delta^\top = (A~B) - (\widehat{A}~\widehat{B})$. For $(A, B) \in E$ we have:
\begin{align*}
    1 &\ge \vectorized(\Delta^\top)^\top \left(\frac{1}{\sigma_w^2c_\delta}\sum_{j=1}^i\Phi_{j-1}^\top\Phi_{j-1} + \frac{\lambda}{ c_\delta\sigma_w^2}I\right)\vectorized(\Delta^\top) \\
    &= \vectorized(\Delta^\top)^\top \left(\left(\frac{1}{\sigma_w^2c_\delta}\sum_{j=1}^iz_{j-1}z_{j-1}^\top+\frac{\lambda}{ c_\delta\sigma_w^2}I_{d_x+d_u}\right)\otimes I_{d_x}\right)\vectorized(\Delta^\top)  \\
    &= \Tr(\Delta^\top D \Delta) \ge \lambda_{\max}(\Delta^\top D \Delta). \\
\end{align*}
where $D = \frac{1}{c_\delta\sigma_w^2}\left(\sum_{j=1}^iz_{j-1}z_{j-1}^\top+\lambda I_{d_x+d_u}\right)$. Hence with probability at least $1-\delta$ matrix $(A_*, B_*)$ lies in the set $\Theta = \{(A, B)| \Delta^\top D \Delta \preceq I, \Delta^\top = (A~B) - (\widehat{A}~\widehat{B})\}$.
\newpage 
\section{Semi Definite Program from SLS}
\label{appendix: Semi Definite Program from SLS}

In the following we will derive a semi-definite constraint which deals with the ellipsoid $\Theta$. To formalize, we denote $\Delta^\top = (A~B)-(\widehat{A}~\widehat{B})$ and would like to solve the following problem:
\begin{equation}
\label{before kyp and s lemma}
    \begin{aligned}
    &\text{find }K \\
    &\forall \Delta \text{ with } \Delta^\top D \Delta \preceq I: \\
    &\qquad\norm{\Delta^\top
    \begin{pmatrix}
    I \\
    K
    \end{pmatrix}
    (zI - \widehat{A} - \widehat{B}K)^{-1}
    }_{\mathcal{H}_\infty} < 1.
    \end{aligned}
\end{equation}

In the problem posed by the \cref{before kyp and s lemma} there are two issues which we need to solve. First there is the condition that we would like to find controller $K$ for which a $\mathcal{H}_\infty$ constraint holds for every $\Delta$ with $\Delta^\top D \Delta \preceq I$. First we will apply the S-lemma (c.f. \citep{LMIs2004}) and obtain a single $\mathcal{H}_\infty$ constraint. Next we will transform the $\mathcal{H}_\infty$ norm constraint to the semi-definite constraint by application of KYP lemma (c.f. \citep{KYPMoreReliableSource}).

\paragraph{First S then KYP lemma}
\label{paragraph:First S then KYP lemma}
The constraint
\begin{align*}
    \norm{\Delta^\top
    \begin{pmatrix}
    I \\
    K
    \end{pmatrix}
    (zI - \widehat{A} - \widehat{B}K)^{-1}
    }_{\mathcal{H}_\infty} < 1
\end{align*}
is equivalent to the constraint that for every $z \in \partial \mathbb{D}$:
\begin{align*}
    \norm{\Delta^\top
    \begin{pmatrix}
    I \\
    K
    \end{pmatrix}
    (zI - \widehat{A} - \widehat{B}K)^{-1}
    }_2 < 1.
\end{align*}
The latter constraint is equivalent to:
\begin{align*}
    \Delta^\top
    \begin{pmatrix}
    I \\
    K
    \end{pmatrix}
    (zI - \widehat{A} - \widehat{B}K)^{-1}(zI - \widehat{A} - \widehat{B}K)^{-\top}\begin{pmatrix}
    I \\
    K
    \end{pmatrix}^\top \Delta \prec I
\end{align*}
which is further equivalent to:
\begin{align*}
    \begin{pmatrix}
    I & (zI - \widehat{A} - \widehat{B}K)^{-\top}\begin{pmatrix}
    I \\
    K
    \end{pmatrix}^\top \Delta \\
    \Delta^\top
    \begin{pmatrix}
    I \\
    K
    \end{pmatrix}
    (zI - \widehat{A} - \widehat{B}K)^{-1} & I
    \end{pmatrix}
    \succ 0.
\end{align*}
The problem given by \cref{before kyp and s lemma} is therefore equivalent to:
\begin{equation}
\label{double forall formulation}
    \begin{aligned}
        &\forall \Delta \text{ with } \Delta^\top D \Delta \preceq I, \forall z \in \partial \mathbb{D}:\\
        &\qquad
        \begin{pmatrix}
    I & (zI - \widehat{A} - \widehat{B}K)^{-\top}\begin{pmatrix}
    I \\
    K
    \end{pmatrix}^\top \Delta \\
    \Delta^\top
    \begin{pmatrix}
    I \\
    K
    \end{pmatrix}
    (zI - \widehat{A} - \widehat{B}K)^{-1} & I
    \end{pmatrix}
    \succ 0.
    \end{aligned}
\end{equation}
By S lemma of \citet{LMIs2004}, \cref{double forall formulation} is further equivalent to:
\begin{equation}
\label{forall t from 0 to infinity}
    \begin{aligned}
    &\forall z \in \partial \mathbb{D}, \exists t \in (0, \infty) \text{ s.t.}: \\
    &\qquad\begin{pmatrix}
    I & 0 & (zI - \widehat{A}- \widehat{B}K)^{-\top}\begin{pmatrix} I \\ K \end{pmatrix}^\top \\
    0 & (1-t)I & 0 \\
    \begin{pmatrix} I \\ K \end{pmatrix}(zI - \widehat{A}-\widehat{B}K)^{-1} & 0 & tD
    \end{pmatrix}
    \succ 0.
    \end{aligned}
\end{equation}
Observe that \cref{forall t from 0 to infinity} is then equivalent to:
\begin{equation}
\label{forall t from 0 to 1}
    \begin{aligned}
    &\forall z \in \partial \mathbb{D}, \exists t \in (0, 1) \text{ s.t.}: \\
    &\qquad\begin{pmatrix}
    I & (zI - \widehat{A}- \widehat{B}K)^{-\top}\begin{pmatrix} I \\ K \end{pmatrix}^\top \\
    \begin{pmatrix} I \\ K \end{pmatrix}(zI - \widehat{A}-\widehat{B}K)^{-1}  & tD
    \end{pmatrix}
    \succ 0.
    \end{aligned}
\end{equation}
Next observe that in \cref{forall t from 0 to 1} if the positive definite constraint holds for one $t$ it will also hold for all $t' \in [t, 1)$. Therefore instead of searching at every $z \in \partial \mathbb{D}$ for suitable $t$ we can equivalently search uniformly in $t$ -- we can take the supremum. Hence \cref{forall t from 0 to 1} is equivalent to:
\begin{equation}
\label{t from 0 to 1 forall z}
    \begin{aligned}
    & \exists t \in (0, 1) \text{ s.t. } \forall z \in \partial \mathbb{D}: \\
    &\qquad\begin{pmatrix}
    I & (zI - \widehat{A}- \widehat{B}K)^{-\top}\begin{pmatrix} I \\ K \end{pmatrix}^\top \\
    \begin{pmatrix} I \\ K \end{pmatrix}(zI - \widehat{A}-\widehat{B}K)^{-1}  & tD
    \end{pmatrix}
    \succ 0.
    \end{aligned}
\end{equation}
Matrix $D$ is positive definite hence $D^{-\frac{1}{2}}$ exists. Observe that conjugating the positive definite constraint in \cref{t from 0 to 1 forall z} with matrix $\diag(I, D^{-\frac{1}{2}})$ the \cref{t from 0 to 1 forall z} is equivalent to:
\begin{equation}
\label{t from 0 to 1 forall z before  norm}
    \begin{aligned}
    & \exists t \in (0, 1) \text{ s.t. } \forall z \in \partial \mathbb{D}: \\
    &\qquad\begin{pmatrix}
    I & \frac{1}{\sqrt{t}}(zI - \widehat{A}- \widehat{B}K)^{-\top}\begin{pmatrix} I \\ K \end{pmatrix}^\top D^{-\frac{1}{2}} \\
    \frac{1}{\sqrt{t}}D^{-\frac{1}{2}}\begin{pmatrix} I \\ K \end{pmatrix}(zI - \widehat{A}-\widehat{B}K)^{-1}  & I
    \end{pmatrix}
    \succ 0,
    \end{aligned}
\end{equation}
which is by using Schur complement lemma further equivalent to:
\begin{equation}
\label{t from 0 to 1 forall z 2 norm}
    \begin{aligned}
    & \exists t \in (0, 1) \text{ s.t. } \forall z \in \partial \mathbb{D}: \\
    &\qquad \norm{\frac{1}{\sqrt{t}}D^{-\frac{1}{2}}\begin{pmatrix} I \\ K \end{pmatrix}(zI - \widehat{A}-\widehat{B}K)^{-1}}_2 < 1.
    \end{aligned}
\end{equation}
By the definition of $\mathcal{H}_\infty$ norm this is further equivalent to:
\begin{equation}
\label{t from 0 to 1 forall z H infinity norm}
    \begin{aligned}
    & \exists t \in (0, 1) \text{ s.t.}: \\
    &\qquad \norm{\frac{1}{\sqrt{t}}D^{-\frac{1}{2}}\begin{pmatrix} I \\ K \end{pmatrix}(zI - \widehat{A}-\widehat{B}K)^{-1}}_{\mathcal{H}_\infty} < 1.
    \end{aligned}
\end{equation}
Now we are in the position to apply discrete KYP lemma of \citet{KYPMoreReliableSource}. This yields that the \cref{t from 0 to 1 forall z H infinity norm} is equivalent to:
\begin{equation}
\label{t from 0 to 1 P postive definite}
    \begin{aligned}
    & \exists t \in (0, 1), \exists X \succ 0 \text{ s.t.}: \\
    &\qquad 
    \begin{pmatrix}
    \widehat{A} + \widehat{B}K & I \\
    \frac{1}{\sqrt{t}}D^{-\frac{1}{2}}\begin{pmatrix} I \\ K \end{pmatrix} & 0
    \end{pmatrix}
    \begin{pmatrix}
    X & 0 \\
    0 & I
    \end{pmatrix}
    \begin{pmatrix}
    \widehat{A} + \widehat{B}K & I \\
    \frac{1}{\sqrt{t}}D^{-\frac{1}{2}}\begin{pmatrix} I \\ K \end{pmatrix} & 0
    \end{pmatrix}^\top
    \preceq
    \begin{pmatrix}
    X & 0 \\
    0 & I
    \end{pmatrix}.
    \end{aligned}
\end{equation}
Applying Schur complement lemma we observe that the positive definite constraint given by \cref{t from 0 to 1 P postive definite} is equivalent to:
\begin{equation}
\label{one matrix inequality step 1}
    \begin{aligned}
    \begin{pmatrix}
    X & 0 & \widehat{A}+\widehat{B}K & I \\
    0 & I & \frac{1}{\sqrt{t}}D^{-\frac{1}{2}}\begin{pmatrix} I \\ K \end{pmatrix} & 0 \\
    (\widehat{A}+\widehat{B}K)^\top & \frac{1}{\sqrt{t}}\begin{pmatrix} I \\ K \end{pmatrix}^\top D^{-\frac{1}{2}} & X^{-1} & 0 \\
    I & 0 & 0 & I
    \end{pmatrix}
    \succeq 0
    \end{aligned}
\end{equation}
Conjugating with matrix $\diag(I, \sqrt{t}I, X, I)$ and denoting $S = KX$ we obtain that \cref{one matrix inequality step 1} is equivalent to:
\begin{equation}
\label{one matrix inequality step 2}
    \begin{aligned}
    \begin{pmatrix}
    X & 0 & \widehat{A}X+\widehat{B}S & I \\
    0 & tI & D^{-\frac{1}{2}}\begin{pmatrix} X \\ S \end{pmatrix} & 0 \\
    (\widehat{A}X+\widehat{B}S)^\top & \begin{pmatrix} X \\ S \end{pmatrix}^\top D^{-\frac{1}{2}} & X & 0 \\
    I & 0 & 0 & I
    \end{pmatrix}
    \succeq 0
    \end{aligned}
\end{equation}
Taking Schur complement lemma again we obtain that the \cref{one matrix inequality step 2} is equivalent to:
\begin{equation}
\label{one matrix inequality step 3}
    \begin{aligned}
    \begin{pmatrix}
    X - I & 0 & \widehat{A}X+\widehat{B}S  \\
    0 & tI & D^{-\frac{1}{2}}\begin{pmatrix} X \\ S \end{pmatrix}\\
    (\widehat{A}X+\widehat{B}S)^\top & \begin{pmatrix} X \\ S \end{pmatrix}^\top D^{-\frac{1}{2}} & X \\
    \end{pmatrix}
    \succeq 0
    \end{aligned}
\end{equation}
Conjugating by matrix
\begin{align*}
 \begin{pmatrix}
    I & 0 & 0 \\
    0 & 0 & I \\
    0 & D^{\frac{1}{2}} & 0
\end{pmatrix}
\end{align*}
we obtain that \cref{one matrix inequality step 3} is further equivalent to:
\begin{equation}
\label{one matrix inequality step 4}
    \begin{aligned}
    \begin{pmatrix}
    X - I & \widehat{A}X+\widehat{B}S & 0  \\
   (\widehat{A}X+\widehat{B}S)^\top  & X & \begin{pmatrix} X \\ S \end{pmatrix}^\top \\
    0 & \begin{pmatrix} X \\ S \end{pmatrix} &tD \\
    \end{pmatrix}
    \succeq 0.
    \end{aligned}
\end{equation}
We derived that the problem given by \cref{before kyp and s lemma} is equivalent to:
\begin{equation}
\label{convex sdp constraint}
    \begin{aligned}
    &\text{find } t \in (0, 1), X \succ 0,S \\
    &\qquad\text{s.t.} 
     \begin{pmatrix}
    X - I & \widehat{A}X+\widehat{B}S & 0  \\
   (\widehat{A}X+\widehat{B}S)^\top & X & \begin{pmatrix} X \\ S \end{pmatrix}^\top \\
    0 & \begin{pmatrix} X \\ S \end{pmatrix} &tD \\
    \end{pmatrix}
    \succeq 0.
    \end{aligned}
\end{equation}
Hence we can solve the convex feasibility problem:
\begin{equation}
\label{convex sdp program}
    \begin{aligned}
    \min_{t \in (0, 1), X \succ 0,S} 0 \\
    &\qquad\text{s.t.} 
     \begin{pmatrix}
    X - I & \widehat{A}X+\widehat{B}S & 0  \\
    (\widehat{A}X+\widehat{B}S)^\top & X & \begin{pmatrix} X \\ S \end{pmatrix}^\top\\
    0 & \begin{pmatrix} X \\ S \end{pmatrix} &tD \\
    \end{pmatrix}
    \succeq 0.
    \end{aligned}
\end{equation}
From the optimal solution of SDP given by \cref{convex sdp program} we obtain the stabilizing controller via $K = SX^{-1}$.

\newpage 
\section{Optimal Infinite Horizon via SDP}
\label{appendix: Optimal Infinite Horizon via SDP}
In this section we first motivate how we can find the optimal infinite horizon cost and controller from a SDP, we further pose the robust variant of the proposed SDP which we transform to a convex SDP using S lemma of \citet{LMIs2004}.

\subsection{Optimal Infinite Horizon Controller via SDP}
\label{[paragraph: Optimal Infinite horizon cost via SDP}
Assuming that the optimal strategy to choose actions is given via fixed measurable function $u_i = f(x_i)$ and further assuming that via this action selection the limit distribution of state exists we obtain that in the limit we have:
\begin{align}
\label{limit distribution}
    x = A_*x + B_*u + w,
\end{align}
where $x$ is the limit distribution of state and by $u = f(x)$. Zero-mean Gaussian noise with covariance matrix $\sigma_w^2I$ is denoted by $w$. In the following we will compute the variance of the \Cref{limit distribution}. 
\begin{align*}
    \Sigma_{xx} = \var(x) = \var(A_*x + B_*u + w) = (A_*~B_*)\Sigma(A_*~B_*)^\top + \sigma_w^2I,
\end{align*}
where we denoted $\Sigma_{xx} = \var(x)$ and $\Sigma = \var\left(\begin{pmatrix}x\\u\end{pmatrix}\right)$. At the same time under the given assumptions the infinite horizon cost is given by:
\begin{align*}
    \lim_{T \to \infty}\E\left[\frac{1}{T} J_T(\pi)\right] &= \E\left[x^\top Q x + u^\top R u\right] \\
    &= \E\left[ \begin{pmatrix}x \\ u\end{pmatrix}^\top \begin{pmatrix}Q & 0 \\ 0 & R\end{pmatrix}  \begin{pmatrix}x \\ u\end{pmatrix}\right] \\
    &= \Tr \left( \begin{pmatrix}Q&0\\ 0&R \end{pmatrix} \E \left[\begin{pmatrix}x \\ u\end{pmatrix}\begin{pmatrix}x \\ u\end{pmatrix}^\top\right]   \right)\\
    &= \Tr \left( \begin{pmatrix}Q&0\\ 0&R \end{pmatrix} \Sigma \right).
\end{align*}
Assembling the results together we see that the optimal infinite horizon cost can be computed from the SDP:
\begin{equation}
\label{optimal infinite hoizon policy SDP}
\begin{aligned}
\min_{\Sigma \succeq 0} \quad
    &\Tr\left(
    \begin{pmatrix}
        Q & 0 \\
        0 & R
    \end{pmatrix}
    \Sigma
    \right)\\
    &\text{s.t.} \quad
    \Sigma_{xx} = (A_* ~ B_*)\Sigma (A_* ~ B_*)^\top  + \sigma_w^2 I,
\end{aligned}
\end{equation}
However as we can see in \Cref{lemma: SDP relaxation gives the same result}, we can also change the equality constraint to semi-definite constraint and obtain the same minimization problem:
\begin{equation}
\label{optimal infinite hoizon policy SDP relaxed}
\begin{aligned}
\min_{\Sigma \succeq 0} \quad
    &\Tr\left(
    \begin{pmatrix}
        Q & 0 \\
        0 & R
    \end{pmatrix}
    \Sigma
    \right)\\
    &\text{s.t.} \quad
    \Sigma_{xx} \succeq (A_* ~ B_*)\Sigma (A_* ~ B_*)^\top  + \sigma_w^2 I,
\end{aligned}
\end{equation}

\begin{lemma}
\label{lemma: SDP relaxation gives the same result}
 For the optimal $\Sigma$ of SDP given by \cref{optimal infinite hoizon policy SDP relaxed} we have:
 \begin{align*}
     \Sigma_{xx} = (A_* ~ B_*)\Sigma (A_* ~ B_*)^\top  + \sigma_w^2 I.
 \end{align*}
\end{lemma}

\begin{proof}
    Assume that
    \begin{align*}
        \Sigma_{xx} = (A_* ~ B_*)\Sigma (A_* ~ B_*)^\top  +  \sigma_w^2 I + E,
    \end{align*}
    where $E \succeq 0$ and $E \ne 0$. Since
    \begin{align*}
        \Sigma_{xx} - E &=  (A_* ~ B_*)\Sigma (A_* ~ B_*)^\top  +  \sigma_w^2 I \\
        &\succeq (A_* ~ B_*)\left(\Sigma - \begin{pmatrix}E & 0 \\ 0 & 0\end{pmatrix}\right) (A_* ~ B_*)^\top  +  \sigma_w^2 I,
    \end{align*}
    also
    \begin{align*}
        \Sigma - \begin{pmatrix}E & 0 \\ 0 & 0\end{pmatrix}
    \end{align*}
    is feasible solution. And since $Q$ is positive semi-definite its cost is smaller than the one of $\Sigma$.
\end{proof}

\subsection{Robust Formulation}

As we will see in the \Cref{lemma: Better covariance matrix for Epsilon region} the semi-definite constraint in \Cref{optimal infinite hoizon policy SDP relaxed} ensures the closed loop stability of system $A_*, B_*$ with controller $K = \Sigma_{ux}\Sigma_{xx}^{-1}$. We will now write the SDP given by \cref{optimal infinite hoizon policy SDP relaxed} in a robust variant:
\begin{equation}
\label{SDP: stabilizing SDP appendix}
    \begin{aligned}
        \min_{\Sigma \succeq 0} &\Tr\left(
    \begin{pmatrix}
        Q & 0 \\
        0 & R
    \end{pmatrix}
    \Sigma
    \right)\\
    &\text{s.t. } \forall (A, B) \in \Theta: ~ 
    \Sigma_{xx} \succeq (A~B)\Sigma(A~B)^\top + \sigma_w^2 I
    \end{aligned}
\end{equation}
Next we show that from any feasible solution $\Sigma$ of the SDP given by \cref{SDP: stabilizing SDP} we can synthesize a controller $K$ which stabilizes every system in $\Theta$.

\begin{lemma}
\label{lemma: Better covariance matrix for Epsilon region}
    Let $\Sigma$ be a feasible solution of SDP given by \cref{SDP: stabilizing SDP appendix}. Then we have:
    \begin{enumerate}
        \item $\Sigma'$ of the form 
        \begin{align*}
            \Sigma' = 
            \begin{pmatrix}
            \Sigma_{xx} & \Sigma_{xx}K^\top \\
            K\Sigma_{xx} & K\Sigma_{xx}K^\top
            \end{pmatrix},
        \end{align*}
        where $K = \Sigma_{ux}\Sigma_{xx}^{-1}$, is also feasible solution of the SDP given by \eqref{SDP: stabilizing SDP} with cost at most that of $\Sigma$.
        \item For $K = \Sigma_{ux}\Sigma_{xx}^{-1}$ we have:  $\forall (A, B) \in \Theta:~ \rho(A + BK) < 1$.
    \end{enumerate}
\end{lemma}
\begin{proof}
Since 
\begin{align*}
    \Sigma - \Sigma' = \begin{pmatrix}0 & 0 \\ 0 & \Sigma_{uu} - \Sigma_{ux}\Sigma_{xx}^{-1}\Sigma_{xu} \end{pmatrix}
\end{align*}
and $\Sigma_{uu} - \Sigma_{ux}\Sigma_{xx}^{-1}\Sigma_{xu}$ is Schur complement of $\Sigma$ we have $\Sigma_{uu} - \Sigma_{ux}\Sigma_{xx}^{-1}\Sigma_{xu} \succeq 0$ and consequently $\Sigma \succeq \Sigma'$. Now fix aribtrary $(A, B) \in \Theta$. We have $\Sigma_{xx} \succeq (A~B)\Sigma(A~B)^\top + \sigma_w^2 I \succeq (A~B)\Sigma'(A~B)^\top + \sigma_w^2 I$, therefore $\Sigma'$ is feasible. Next we will show $\rho(A+BK)<1$. The semi-definite inequality
\begin{align*}
    \Sigma_{xx} \succeq (A~B)\Sigma'(A~B)^\top + \sigma_w^2 I
\end{align*}
is equivalent to:
\begin{align*}
    \Sigma_{xx} \succeq (A +BK)\Sigma_{xx}(A+BK)^\top + \sigma_w^2 I.
\end{align*}
Let $\mu, v$ be eigenpair of $(A+BK)^\top$. We have:
\begin{align*}
    v^{H}\Sigma_{xx}v \ge \abs{\mu}^2v^H\Sigma_{xx}v + \sigma_w^2\norm{v}^2 > \abs{\mu}^2v^H\Sigma_{xx}v.
\end{align*}
Hence $\abs{\mu} < 1$.
\end{proof}

In the following we will rewrite SDP given by \cref{SDP: stabilizing SDP} to a convex SDP using S lemma of \citet{LMIs2004}. Inserting $(A~B) = \Delta^\top + (\widehat{A}~\widehat{B})$ to \cref{SDP: stabilizing SDP} we obtain that SDP given by \cref{SDP: stabilizing SDP} is equivalent to:
\begin{equation}
\label{SDP: stabilizing SDP with X}
    \begin{aligned}
        \min_{\Sigma \succeq 0} &\Tr\left(
    \begin{pmatrix}
        Q & 0 \\
        0 & R
    \end{pmatrix}
    \Sigma
    \right)\\
    &\text{s.t. } \forall (A, B) \in \Theta: \\
    &
    \Sigma_{xx} - \sigma_w^2 I - \Delta^\top \Sigma \Delta - \Delta^\top \Sigma (\widehat{A}~\widehat{B})^\top - (\widehat{A}~\widehat{B})\Sigma \Delta - (\widehat{A}~\widehat{B})\Sigma (\widehat{A}~\widehat{B})^\top  \succeq 0
    \end{aligned}
\end{equation}
The latter is by S lemma equivalent to:
\begin{equation}
\label{SDP: stabilizing SDP final appendix}
    \begin{aligned}
        \min_{\Sigma \succeq 0, t \ge 0} &\Tr\left(
    \begin{pmatrix}
        Q & 0 \\
        0 & R
    \end{pmatrix}
    \Sigma
    \right)\\
    &\text{s.t. } 
    \begin{pmatrix}
        \Sigma_{xx} - (\widehat{A}~\widehat{B})\Sigma (\widehat{A}~\widehat{B})^\top - (t + \sigma_w^2) I & (\widehat{A}~\widehat{B})\Sigma \\
        \Sigma (\widehat{A}~\widehat{B})^\top & t D - \Sigma
    \end{pmatrix}\succeq 0.
    \end{aligned}
\end{equation}
This is a convex formulation of SDP and we can solve it using e.g. MOSEK \citep{mosek}.
\newpage 
\section{Proof of \texorpdfstring{\Cref{theorem: feasibility equivalence}}{Theorem 1}}
\label{appendix: Proof of Theorem 1}

Even though the way we obtained SDP \eqref{semi-definite constraint} and SDP \eqref{SDP: stabilizing SDP final} are different we show in this section that in fact they are the same, meaning that as soon as one SDP is feasible the other is feasible as well. To see this first note that semi-definite constraint in \cref{semi-definite constraint} can be rewritten as:
\begin{align}
\label{SDP to SLS step 1}
    \begin{pmatrix}
    X - I & (\widehat{A}~\widehat{B})\begin{pmatrix}I \\ K\end{pmatrix}X & 0 \\
    X\begin{pmatrix}I \\ K\end{pmatrix}^\top (\widehat{A}~\widehat{B})^\top & X & X\begin{pmatrix}I \\ K\end{pmatrix}^\top \\
    0 & \begin{pmatrix}I \\ K\end{pmatrix}X & tD
    \end{pmatrix} \succeq 0.
\end{align}
Conjugating by matrix
\begin{align*}
    \begin{pmatrix}
    I & 0 & 0 \\
    0 & 0 & I \\
    0 & I & 0
    \end{pmatrix}
\end{align*}
we obtain that \cref{SDP to SLS step 1} is equivalent to:
\begin{align}
\label{SDP to SLS step 2}
    \begin{pmatrix}
    X - I & 0 & (\widehat{A}~\widehat{B})\begin{pmatrix}I \\ K\end{pmatrix}X \\
    0 & tD &  \begin{pmatrix}I \\ K\end{pmatrix}X\\
    X\begin{pmatrix}I \\ K\end{pmatrix}^\top (\widehat{A}~\widehat{B})^\top & X\begin{pmatrix}I \\ K\end{pmatrix}^\top & X
    \end{pmatrix} \succeq 0.
\end{align}
We can rewrite \cref{SDP to SLS step 2} using Schur complements lemma to:
\begin{align}
    \label{SDP to SLS step 3}
    \begin{pmatrix}
    X - I & 0 \\
    0 & tD
    \end{pmatrix}
    - 
    \begin{pmatrix}
    (\widehat{A}~\widehat{B})\begin{pmatrix}I \\ K\end{pmatrix}X \\
    \begin{pmatrix}I \\ K\end{pmatrix}X\\
    \end{pmatrix}
    X^{-1}
    \begin{pmatrix}
     X\begin{pmatrix}I \\ K\end{pmatrix}^\top (\widehat{A}~\widehat{B})^\top & X\begin{pmatrix}I \\ K\end{pmatrix}^\top
    \end{pmatrix}
    \succeq 0,
\end{align}
which is, by multiplying the matrices, further equivalent to:
\begin{align}
    \label{SDP to SLS step 4}
    \begin{pmatrix}
    X - (\widehat{A}~\widehat{B})\begin{pmatrix}I \\ K\end{pmatrix}X\begin{pmatrix}I \\ K\end{pmatrix}^\top (\widehat{A}~\widehat{B})^\top - I & -(\widehat{A}~\widehat{B})\begin{pmatrix}I \\ K\end{pmatrix}X\begin{pmatrix}I \\ K\end{pmatrix}^\top \\
    -\begin{pmatrix}I \\ K\end{pmatrix}X\begin{pmatrix}I \\ K\end{pmatrix}^\top(\widehat{A}~\widehat{B})^\top & tD - \begin{pmatrix}I \\ K\end{pmatrix}X\begin{pmatrix}I \\ K\end{pmatrix}^\top
    \end{pmatrix}
    \succeq 0
\end{align}
We know by \Cref{lemma: Better covariance matrix for Epsilon region} that the optimal solution of SDP \eqref{SDP: stabilizing SDP final} is parametrized as 
\begin{align*}
    \Sigma = 
    \begin{pmatrix}
    \Sigma_{xx} & \Sigma_{xx}K^\top \\
    K \Sigma_{xx} & K \Sigma_{xx} K^\top
    \end{pmatrix}
    = 
    \begin{pmatrix}I \\ K\end{pmatrix}\Sigma_{xx}\begin{pmatrix}I \\ K\end{pmatrix}^\top
\end{align*}
Hence by denoting $U = \begin{pmatrix}I \\ K\end{pmatrix}X\begin{pmatrix}I \\ K\end{pmatrix}^\top$ we obtain that \cref{SDP to SLS step 4} can be rewritten as:
\begin{align}
    \label{SDP to SLS step 5}
    \begin{pmatrix}
    U_{xx} - (\widehat{A}~\widehat{B})U (\widehat{A}~\widehat{B})^\top - I & -(\widehat{A}~\widehat{B})U \\
    -U(\widehat{A}~\widehat{B})^\top & tD - U
    \end{pmatrix}
    \succeq 0,
\end{align}
which is further equivalent to:
\begin{align}
    \label{SDP to SLS step 6}
    \begin{pmatrix}
    U_{xx} - (\widehat{A}~\widehat{B})U (\widehat{A}~\widehat{B})^\top - I & (\widehat{A}~\widehat{B})U \\
    U(\widehat{A}~\widehat{B})^\top & tD - U
    \end{pmatrix}
    \succeq 0
\end{align}
To show that SDP \eqref{semi-definite constraint} is feasible if and only if SDP \eqref{SDP: stabilizing SDP final} is feasible is then equivalent to show that 
\begin{equation}
\begin{aligned}
\label{SLS part of equivalence}
    \exists U \succeq 0, s \in (0, 1) &\text{ s.t.:}\\
    &\begin{pmatrix}
    U_{xx} - (\widehat{A}~\widehat{B})U (\widehat{A}~\widehat{B})^\top - I & (\widehat{A}~\widehat{B})U \\
    U(\widehat{A}~\widehat{B})^\top & sD - U
    \end{pmatrix}
    \succeq 0
\end{aligned}
\end{equation}
is equivalent to:
\begin{equation}
\begin{aligned}
\label{SDP part of equivalence}
    \exists \Sigma \succeq 0, t \ge 0 &\text{ s.t.:}\\
    &\begin{pmatrix}
    \Sigma_{xx} - (\widehat{A}~\widehat{B})\Sigma (\widehat{A}~\widehat{B})^\top - (t + \sigma_w^2)I & (\widehat{A}~\widehat{B})\Sigma \\
    \Sigma(\widehat{A}~\widehat{B})^\top & tD - \Sigma
    \end{pmatrix}
    \succeq 0
\end{aligned}
\end{equation}
Assume that we have \cref{SLS part of equivalence}. Multiply semi-definite constraint in \cref{SLS part of equivalence} with $\frac{\sigma_w^2}{1-s}$ and denote $t = \frac{s\sigma_w^2}{1-s}, \Sigma = \frac{\sigma_w^2}{1-s}U$. With such a notation we have:
\begin{equation}
\begin{aligned}
    \begin{pmatrix}
    \Sigma_{xx} - (\widehat{A}~\widehat{B})\Sigma (\widehat{A}~\widehat{B})^\top - (t + \sigma_w^2)I & (\widehat{A}~\widehat{B})\Sigma \\
    \Sigma(\widehat{A}~\widehat{B})^\top & tD - \Sigma
    \end{pmatrix}
    \succeq 0.
\end{aligned}
\end{equation}
Since $\Sigma = \frac{\sigma_w^2}{1-s}U \succeq 0$ and $t = \frac{s\sigma_w^2}{1-s} \ge 0$ we see that condition given by \cref{SDP part of equivalence} is satisfied. To show the equivalence in other direction assume that we have \cref{SDP part of equivalence}. Multiplying semi-definite constraint in \cref{SDP part of equivalence} with $\frac{1}{t + \sigma_w^2}$ and denoting $s = \frac{t}{t+\sigma_w^2}, U = \frac{1}{t + \sigma_w^2}\Sigma$ we obtain:
\begin{equation}
\begin{aligned}
    \begin{pmatrix}
    U_{xx} - (\widehat{A}~\widehat{B})U (\widehat{A}~\widehat{B})^\top - I & (\widehat{A}~\widehat{B})U \\
    U(\widehat{A}~\widehat{B})^\top & sD - U
    \end{pmatrix}
    \succeq 0.
\end{aligned}
\end{equation}
Since $U = \frac{1}{t + \sigma_w^2}\Sigma \succeq 0$ and $s = \frac{t}{t + \sigma_w^2} < 1$ we obtain that \cref{SLS part of equivalence} is satisfied. Hence we obtained that as soon as one of the SDP \cref{semi-definite constraint} or SDP \cref{SDP: stabilizing SDP final} is feasible, the other is feasible as well.
\newpage 
\section{Initialization of Existing Algorithms}
\label{appendix: Initialization of Existing Algorithms}

 We have seen how we can find a controller which with high probability stabilizes the system $A_*, B_*$ in time which depends only on the system parameters. Here we will show how we can initialize the existing algorithms, such as \textsc{OSLO} \citep{Cohen2019SPDRelaxation} or \textsc{CEC} \citep{simchowitz2020naive}, which require a stabilizing controller as an input, with \textsc{eXploration}. Both algorithms, \textsc{OSLO} and \textsc{CEC}, consist of two parts. In the first part, which we call \emph{warm up phase}, they utilize the stabilizing controller to obtain tight estimates of system matrices $A_*, B_*$, which knowledge they utilize in the second part, where they choose actions optimistically (\textsc{OSLO}) or greedily (\textsc{CEC}). Together with \textsc{eXploration} as initialization we obtain two 3-phased algorithms which we call \textsc{X-OSLO} and \textsc{X-CEC}. 

The second phase of \textsc{X-OSLO} and \textsc{X-CEC} is given in \cref{algorithm: utilize the stabilizing controller}. Parameter $\sigma_{init}^2$ is different for both algorithms, also the number of steps we run \cref{algorithm: utilize the stabilizing controller} differs between \textsc{OSLO} and \textsc{CEC}.

\begin{algorithm}[H]
    \caption{Utilize the stabilizing controller}
    \label{algorithm: utilize the stabilizing controller}
    \begin{algorithmic}[1]
        \STATE {\textbf{Input:}}{ Controller $K$ with $\rho(A_*+B_*K)<1$}
        \FOR{$i=1, \ldots $}{
            {\STATE {\textbf{observe} state $x_i$}}
            {\STATE {\textbf{play} $u_i\sim \mathcal{N}(Kx_i, \sigma_{init}^2 I)$}}
            }
        \ENDFOR
    \end{algorithmic}
\end{algorithm}

\subsection{Initialization of \textsc{OSLO}}
\label{section: Initialization of OSLO}

In the second phase of \textsc{X-OSLO} we set $\sigma_{init}^2 = 2\sigma_w^2\kappa_0^2$, where $\kappa_0$ is the first of the so called \emph{strongly stable} (c.f. \citep{Cohen2018}) parameters of the controller $K$. 

\begin{definition}
\label{definition: Strong stability}
    A controller $K$ is $(\kappa, \gamma)$-strongly stable for $0 < \gamma \le 1$ if:
    \begin{enumerate}
        \item $\norm{K}_2 \le \kappa$
        \item $A_* + B_*K = HLH^{-1}$, with $\norm{L}_2 \le 1-\gamma$ and $\norm{H}_2\norm{H^{-1}}_2 \le \kappa$.
    \end{enumerate}
    Here we call $\kappa$ and $\gamma$ the first and the second strongly stable parameter respectively.
\end{definition}

In the rest of this section we will show we can obtain strongly stable parameter from the controller which we obtain from SDP \eqref{semi-definite constraint} or SDP \eqref{SDP: stabilizing SDP final}.

\paragraph{Strong stability parameters from robust SDP}
\label{subsection: Strong stability parameters from SDP}
In the following we denote $K = \Sigma_{ux}\Sigma_{xx}^{-1}$, where $\Sigma$ is the optimal solution of SDP \eqref{SDP: stabilizing SDP final}. 

\begin{lemma}
    Assume we synthesize a controller with SDP \eqref{SDP: stabilizing SDP final} and let $\Sigma$ be the optimal solution of SDP \eqref{SDP: stabilizing SDP final}. Denote by $\nu= \Tr(\Sigma)$ and $\kappa^2=\frac{\nu}{\sigma_w^2}$ then controller $K$ is $(\kappa, \frac{1}{2\kappa^2})$-strongly stable.
\end{lemma}

\begin{proof}
Since for every $(A, B) \in \Theta$ (also for $(A_*, B_*)$) we have $\Sigma_{xx}\succeq (A+BK)\Sigma_{xx}(A+BK)^\top + \sigma_w^2I$ we have $\sigma_w^2I \preceq \Sigma_{xx}$. Since we know $\Sigma$, we can compute its trace $\nu = \Tr(\Sigma_{xx}) + \Tr(\Sigma_{uu})$. With such a notation we have:
$\sigma_w^2I \preceq \Sigma_{xx} \preceq \nu I$. Denote by $L = \Sigma_{xx}^{-1/2}(A_*+B_*K)\Sigma_{xx}^{1/2}$. Multiplying equation
\begin{align*}
    \Sigma_{xx} \succeq (A_*+B_*K)\Sigma_{xx}(A_*+B_*K)^\top + \sigma_w^2 I
\end{align*}
from left and right with $\Sigma_{xx}^{-1/2}$ we obtain:
\begin{align*}
    I \succeq LL^\top +\sigma_w^2\Sigma_{xx}^{-1} \succeq LL^\top + \frac{\sigma_w^2}{\nu}I.
\end{align*}
From there it follows:
\begin{align*}
    LL^\top \preceq \left(1 - \frac{\sigma_w^2}{\nu}\right)I,
\end{align*}
which yields: $\norm{L}_2 \le \sqrt{1-1/\kappa^2}\le 1 - \frac{1}{2\kappa^2}$.
In the notation of \Cref{definition: Strong stability} we have $H = \Sigma_{xx}^{1/2}$. Since $\sigma_w^2I \preceq \Sigma_{xx} \preceq \nu I$ we have: $\norm{\Sigma_{xx}^{1/2}}_2\norm{\Sigma_{xx}^{-1/2}}_2 \le \sqrt{\nu}\frac{1}{\sigma_w} = \kappa$. To finish the proof observe:
\begin{align*}
    \sigma_w^2 \norm{K}_F^2 \le \Tr(K\Sigma_{xx}K^\top) = \Tr(\Sigma_{uu}) \le \nu,
\end{align*}
from where we conclude: $\norm{K}_2 \le \norm{K}_F \le \kappa$.
\end{proof}

From the discussion in \Cref{appendix: Proof of Theorem 1} we see that we can obtain strong stability parameters also from the solution of SDP \eqref{semi-definite constraint}. If we define 
\begin{align*}
    \Sigma' = \frac{\sigma_w^2}{1-t}\begin{pmatrix}I \\K\end{pmatrix}P\begin{pmatrix}I \\K\end{pmatrix}^\top, \quad t' = \frac{t\sigma_w^2}{1-t}   
\end{align*}
then from the reformulation of SDP \eqref{semi-definite constraint} given in \Cref{appendix: Proof of Theorem 1} follows that $\Sigma', t'$ are feasible solution of SDP \eqref{SDP: stabilizing SDP final} and hence the following lemma holds:

\begin{lemma}
    Let $P, K, t$ be the parameters of the optimal soluton of SDP \eqref{semi-definite constraint}. Then for $\kappa^2 = \frac{1}{1-t}\Tr(P(I+K^\top K))$ controller $K$ is $(\kappa, \frac{1}{2\kappa^2})$ strongly stable.
\end{lemma}

Since the cost suffered during the run of \textsc{eXploration} is constant in $T$, albeit could be exponentially large in systems parameters (e.g. $\norm{A_*}$), we obtain the following theorem.

\begin{theorem}
Suppose the system matrices $A_*, B_*$ are stabilizable and regular, cost matrices $Q, R$ are positive definite and time horizon is $T$. Then by first running \textsc{eXploration}, where we synthesize the controller with Robust SLS \eqref{subsection: Robust formulation from System Level Synthesis} or Robust LQR \eqref{subsection: Robust formulation from Semi-Definite Program}, using data dependent upper bounds from the Bayesian setting, and then \textsc{OSLO} algorithm, the total regret we suffer is upper bounded with probability at least $1-\delta$ as:
\begin{align*}
    R(T, \textsc{X-OSLO}) = \mathcal{O}\left(\sqrt{T}\log^2T\right).
\end{align*}
\end{theorem}

\subsection{Initialization of \textsc{CEC}}
\label{section: Initialization of CEC}

Initialization of \textsc{CEC} requires only the stabilizing controller $K$. Hence we can directly state the theorem.

\begin{theorem}
Suppose the system matrices $A_*, B_*$ are stabilizable and regular, cost matrices $Q, R$ are positive definite, time horizon is $T$ and probability of failure is $\delta \in (0, \frac{1}{T})$. Then by first running \textsc{eXploration}, where we synthesize the controller with Robust SLS \eqref{subsection: Robust formulation from System Level Synthesis} or Robust LQR \eqref{subsection: Robust formulation from Semi-Definite Program}, using data dependent upper bounds from the Bayesian setting, and then \textsc{CEC} algorithm, the total regret we suffer is upper bounded with probability at least $1-\delta$ as:
\begin{align*}
    R(T, \textsc{X-CEC}) = \mathcal{O}\left(\sqrt{T\log T}\right).
\end{align*}
\end{theorem}

The proof follows directly from the Theorem 2 of \citet{simchowitz2020naive} and \Cref{theorem: Algorithm termination}.
\newpage 
\section{Proof of \texorpdfstring{\Cref{theorem: Algorithm termination}}{Theorem 2}}
\label{appendix: Proof of Theorem 2}

With the notation $\normno{A_* - \widehat{A}}_2 \le \varepsilon_A, \normno{B_* - \widehat{B}}_2 \le \varepsilon_B$ \citet{Dean2018OnSampleComplexity} proved the following lemma: 

\begin{lemma}[Fulfilled Sufficient condition, Lemma 4.2 in \cite{Dean2018OnSampleComplexity}]
    \label{lemma:SufficientConditionToHaveFeasibleSolution}
    Let $K$ be a controller which stabilizes $(A_*, B_*)$.
    Assume that $\varepsilon_A, \varepsilon_B$ are small enough that for $\zeta$ defined as $\zeta = \left(\epsilon_A + \epsilon_B\norm{K}_2\right)\norm{(zI - A_*-B_*K)^{-1}}_{\mathcal{H}_\infty}$ we have $\zeta \le (1 + \sqrt{2})^{-1}$. Then $K$ satisfies the constraint given by \cref{equation:HInfinityConditionToStabilize}.
\end{lemma}

From \cref{lemma:SufficientConditionToHaveFeasibleSolution} follows that if $(\varepsilon_A+\varepsilon_B\norm{K}_2)\norm{\mathfrak{R}_{A_*+B_*K}}_{\mathcal{H_\infty}} \le (1+\sqrt{2})^{-1}$ then the feasibility problem given by \cref{semi-definite constraint} has a solution. Denote by $\varepsilon = \varepsilon_A \lor \varepsilon_B$ and observe that for the systems $(A, B)\in \Theta$ we have $\normno{A - \widehat{A}}_2 \lor \normno{B - \widehat{B}}_2 \le \varepsilon = \frac{1}{\sqrt{\lambda_{\min}(D)}}$. Then a sufficient condition for feasibility of SDP given by \cref{convex sdp constraint} is:
\begin{align*}
    \frac{1}{\sqrt{\lambda_{\min}(D)}} \le \frac{1}{(1+\norm{K}_2)\norm{(zI - A_*-B_*K)^{-1}}_{\mathcal{H}_\infty}(1+\sqrt{2})},
\end{align*}
which is equivalent to:
\begin{align}
\label{sufficient condition from Dean}
    \left((1+\sqrt{2})(1+\norm{K}_2)\norm{(zI - A_*-B_*K)^{-1}}_{\mathcal{H}_\infty}\right)^2 \le \lambda_{\min}(D).
\end{align}

Next we will show that with probability at least $1 - \delta$ the smallest eigenvalue of Gramian matrix $D$ grows linearly with time. Due to the ease of the exposition we will neglect logarithmic terms. Since $D = \frac{1}{\sigma_w^2c_\delta}\left(\sum_{j=1}^iz_{j-1}z_{j-1}^\top + \lambda I\right)$ it is enough to show that the smallest eigenvalues of $V_i = \sum_{j=1}^iz_{j-1}z_{j-1}^\top$ grows linearly. From the analysis of \citet{Sarkar2018Identification} (e.g. Equation 116) follows that for any fixed regular system $A_*, B_*$ for the associated matrix $V_i$ exist a constant $C_{A_*, B_*}$, that we have: $C_{A_*, B_*} i \le \lambda_{min}(V_i)$ w.p. $1-\delta$, where the probability is taken over the randomness of noise $(w_{j})_{j\ge 1}$. In our setting also the system $A_*, B_*$ is random. We choose $A_*, B_*$ such that every entry is independently chosen from a normal distribution. We would like to say that with respect to the randomness of $A_*, B_*$ and $(w_{j})_{j\ge 1}$ with probability $1 -  \delta$ we have: $C i \le \lambda_{min}(V_i)$. Here $C$ would be a constant which depends on $\delta$.

In order to be able to say such statement we construct mapping $Y: (A_*, B_*) \mapsto C_{A_*, B_*}$. Since irregular systems have Lebesgue measure 0, we have $C_{A_*, B_*} > 0$ a.s. Hence we have:
\begin{align}
\label{irregular systems have measure zero}
    \lim_{C \to 0} \pr(Y(A_*, B_*) < C) = 0.
\end{align}
Therefore from \Cref{irregular systems have measure zero} follows that there exists constant $C$ (which depends on $\delta$) s.t.:
\begin{align}
\label{bad event for constant C}
    \pr(Y(A_*, B_*) < C) < \delta.
\end{align}
Now we compute the probability that $C i \le \lambda_{min}(V_i)$ also with respect to randomness of $A_*, B_*$:
\begin{align*}
    \pr(\lambda_{min}(V_i) < Ci) &= \pr(\lambda_{min}(V_i) < Ci, Y < C) + \pr(\lambda_{min}(V_i) < Ci, Y \ge C) \\
    &\le \pr (Y < C) + \pr(\lambda_{min}(V_i) < Ci, Y \ge C) \le \delta + \delta = 2\delta.
\end{align*}
To see that $\pr(\lambda_{min}(V_i) < Ci, X \ge C) \le \delta$ we compute:
\begin{align*}
    \pr(\lambda_{min}(V_i) < Ci, Y \ge C) &= \int \pr(\lambda_{min}(V_i) < Ci, Y \ge C|A_*, B_*)d\mu(A_*, B_*) \\
    & \le \int \delta d\mu(A_*, B_*) = \delta.
\end{align*}
Rescaling $\delta$  we arrive at the result: with probability $1 -  \delta$ we have: $C i \le \lambda_{min}(V_i)$. Applying the result to condition given by \Cref{sufficient condition from Dean} we obtain that \Cref{algorithm:FirstPart} will terminate with probability $1-\delta$ in time:
\begin{align*}
    \mathcal{O}\left((1 + \norm{K}_2)^2\norm{(zI - A_*-B_*K)^{-1}}_{\mathcal{H}_\infty}^2\right)
\end{align*}
\newpage
\section{Convex SDP formulation for MinMax Problem}
\label{appendix: Convex SDP formulation for MinMax}

First note that min max problem can be formulated as:
\begin{equation}
\label{SDP: minimize the spectral norm raw form}
    \begin{aligned}
        \min_{t \ge 0,K} t \\
    &\text{s.t. } \forall (A~B) \in \Theta: \quad
        \norm{A+BK}_2 \le t.
    \end{aligned}
\end{equation}

Next we will transform problem \cref{SDP: minimize the spectral norm raw form} to convex SDP using S lemma of \citet{LMIs2004}. To reformulate the problem in such a way observe first that the constraint $\norm{A+BK}_2 \le t$ can be rewritten using Schur complement lemma as:
\begin{align*}
    &\norm{A+BK}_2 \le t \\
    \iff
    &(A+BK)^\top(A+BK) \preceq t^2I \\
    \iff
    &
    \begin{pmatrix}
    tI & (A+BK)^\top \\
    A+BK & tI
    \end{pmatrix} \succeq 0
\end{align*}
Using the notation from the definition of high probability region given we reformulate the minimization problem given by \cref{SDP: minimize the spectral norm raw form} to:
\begin{equation}
\label{SDP: minimize the spectral norm X form}
    \begin{aligned}
        \min_{t \ge 0,K} t \\
    &\text{s.t. } \forall \Delta \text{ with } \Delta^\top D\Delta \preceq I: \\
    &\quad
        \begin{pmatrix}
        tI & (\widehat A + \widehat B K)^\top - \begin{pmatrix}I \\K \end{pmatrix}^\top \Delta \\
        \left((\widehat A + \widehat B K)^\top - \begin{pmatrix}I \\K \end{pmatrix}^\top \Delta\right)^\top & tI
        \end{pmatrix}
        \succeq 0
    \end{aligned}
\end{equation}
Applying S lemma we obtain that the \cref{SDP: minimize the spectral norm X form} is equivalent to:
\begin{equation}
\label{SDP: minimize the spectral norm inside hp region}
    \begin{aligned}
        \min_{t \ge 0,\lambda \ge 0,  K} t \\
    &\text{s.t. } 
        \begin{pmatrix}
            tI & (\widehat A + \widehat B K)^\top &             \begin{pmatrix}
            I \\ K
            \end{pmatrix}^\top \\
            \widehat A + \widehat B K & (t-\lambda) I & 0 \\
            \begin{pmatrix}
            I \\ K
            \end{pmatrix}
            & 0 & \lambda D
        \end{pmatrix}
        \succeq 0
    \end{aligned}
\end{equation}
which is a convex SDP. At the same time we can also use SDP \eqref{SDP: minimize the spectral norm inside hp region} to bound for a given controller $K'$ the norm of associated closed loop matrix:
\begin{equation}
\label{SDP: bound the controller performance}
    \begin{aligned}
        \min_{t \ge 0,\lambda \ge 0} t \\
    &\text{s.t. } 
        \begin{pmatrix}
            tI & (\widehat A + \widehat B K')^\top &             \begin{pmatrix}
            I \\ K'
            \end{pmatrix}^\top \\
            \widehat A + \widehat B K' & (t-\lambda) I & 0 \\
            \begin{pmatrix}
            I \\ K'
            \end{pmatrix}
            & 0 & \lambda D
        \end{pmatrix}
        \succeq 0
    \end{aligned}
\end{equation}
For the optimal $t$ which we obtain from the solution of SDP \eqref{SDP: bound the controller performance} we have that with probability at least $1-\delta$:
\begin{align*}
    \norm{A_* + B_*K'}_2 \le t.
\end{align*}
\newpage
\section{Additional Experiments}
\label{section: Additional Experiments}

First we provide sample run on the system given by \Cref{equation:SystemFromDean2017}. As we can see on \Cref{fig: ellipsoid ball comparison} if we use ellipsoid bounds we find the stabilizing controller sooner and we suffer considerably less cost. Further we see that with different probing we reduce the initial blow up. How to choose actions optimally during the \textsc{eXploration} is left for future work.

\begin{figure}[ht]
    \begin{subfigure}[$K_i = 0$]{
        \includegraphics[width=0.5\textwidth]{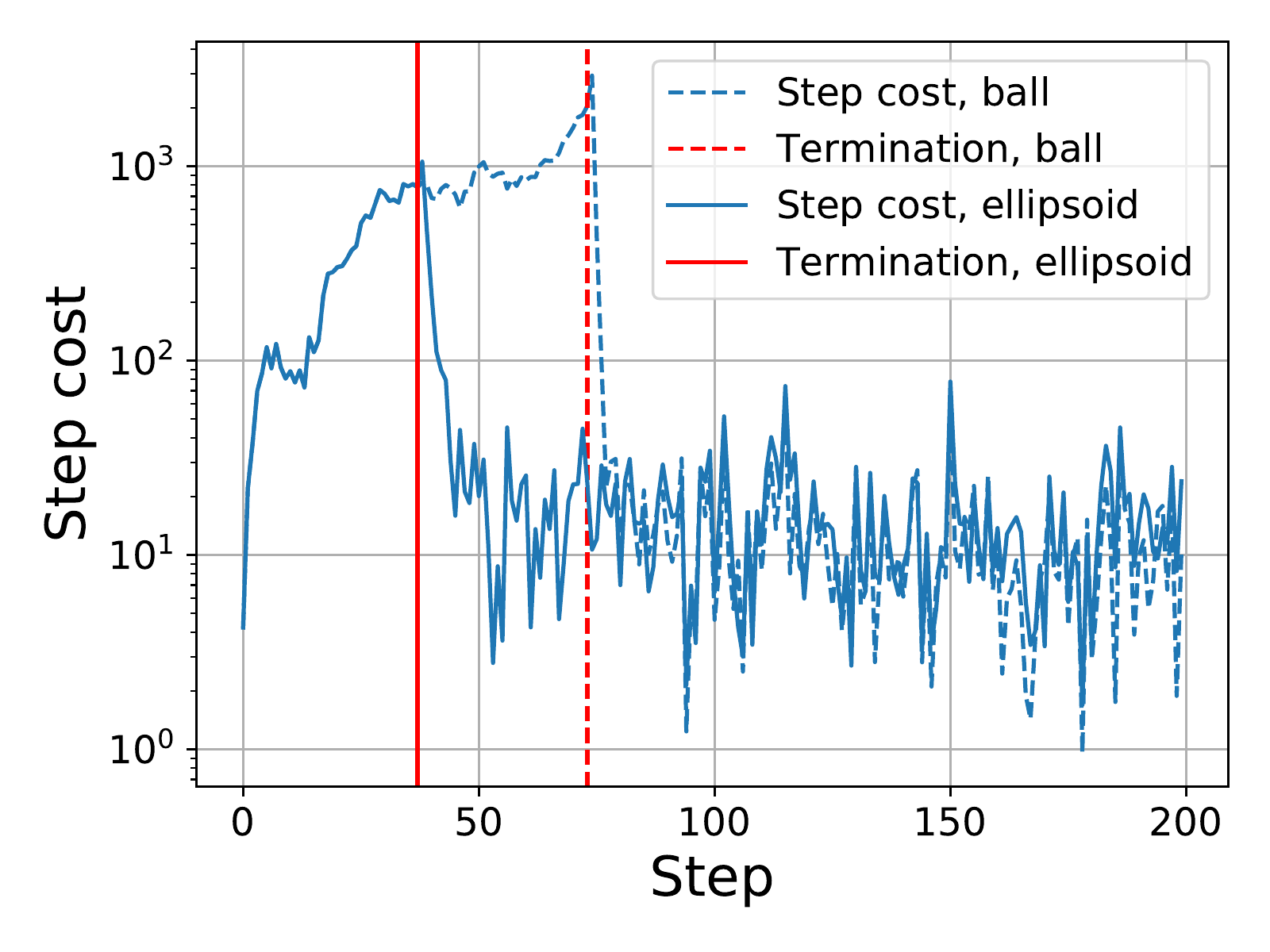}
        \label{fig: ellipsoid ball comparison}
    }
    \end{subfigure}
    \begin{subfigure}[$K_i$ as CE controller]{
        \includegraphics[width=0.5\textwidth]{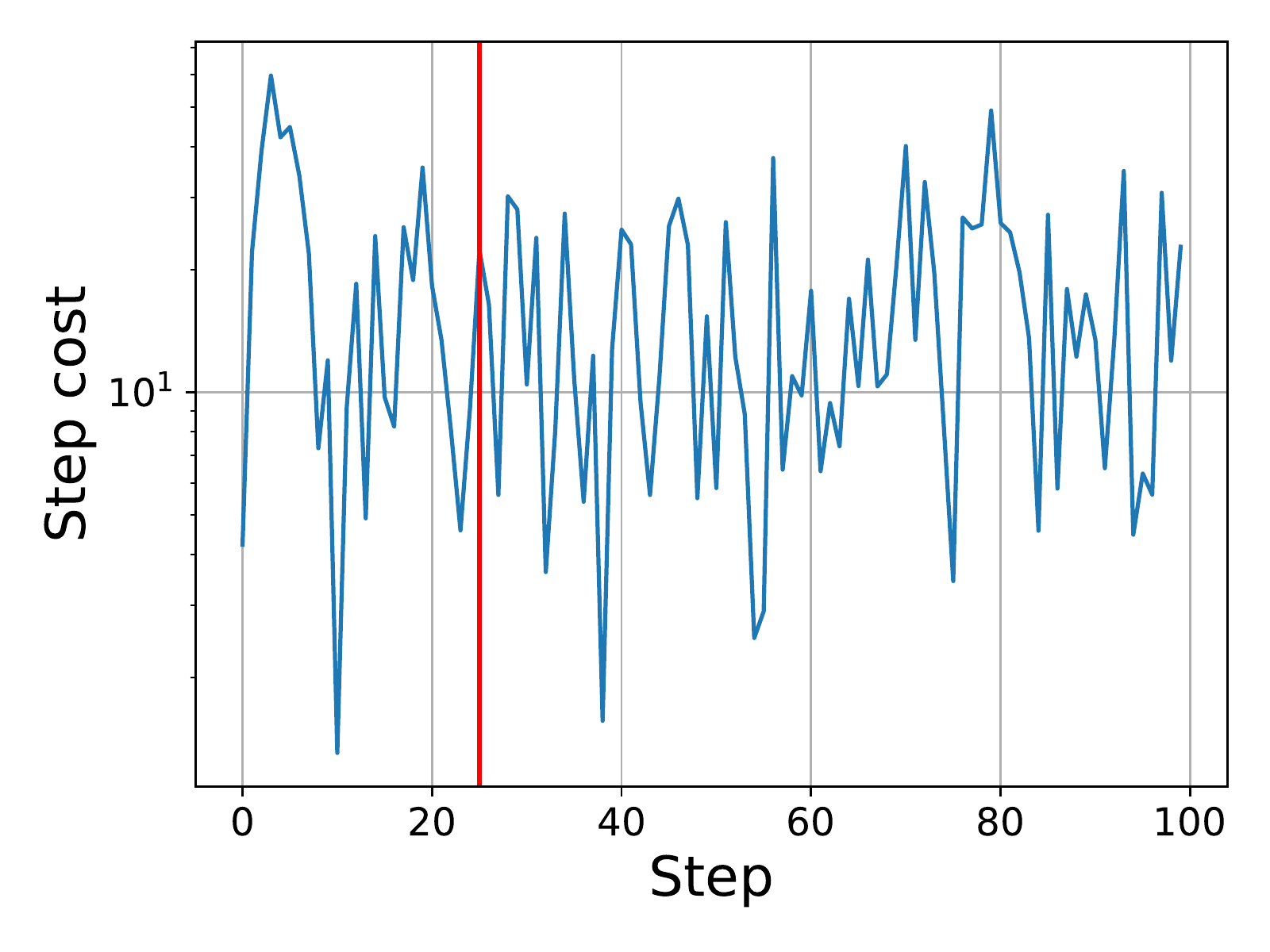}
    }  
    \end{subfigure}
    \begin{subfigure}[$K_i$ as \textsc{MinMax} controller]{
        \includegraphics[width=0.5\textwidth]{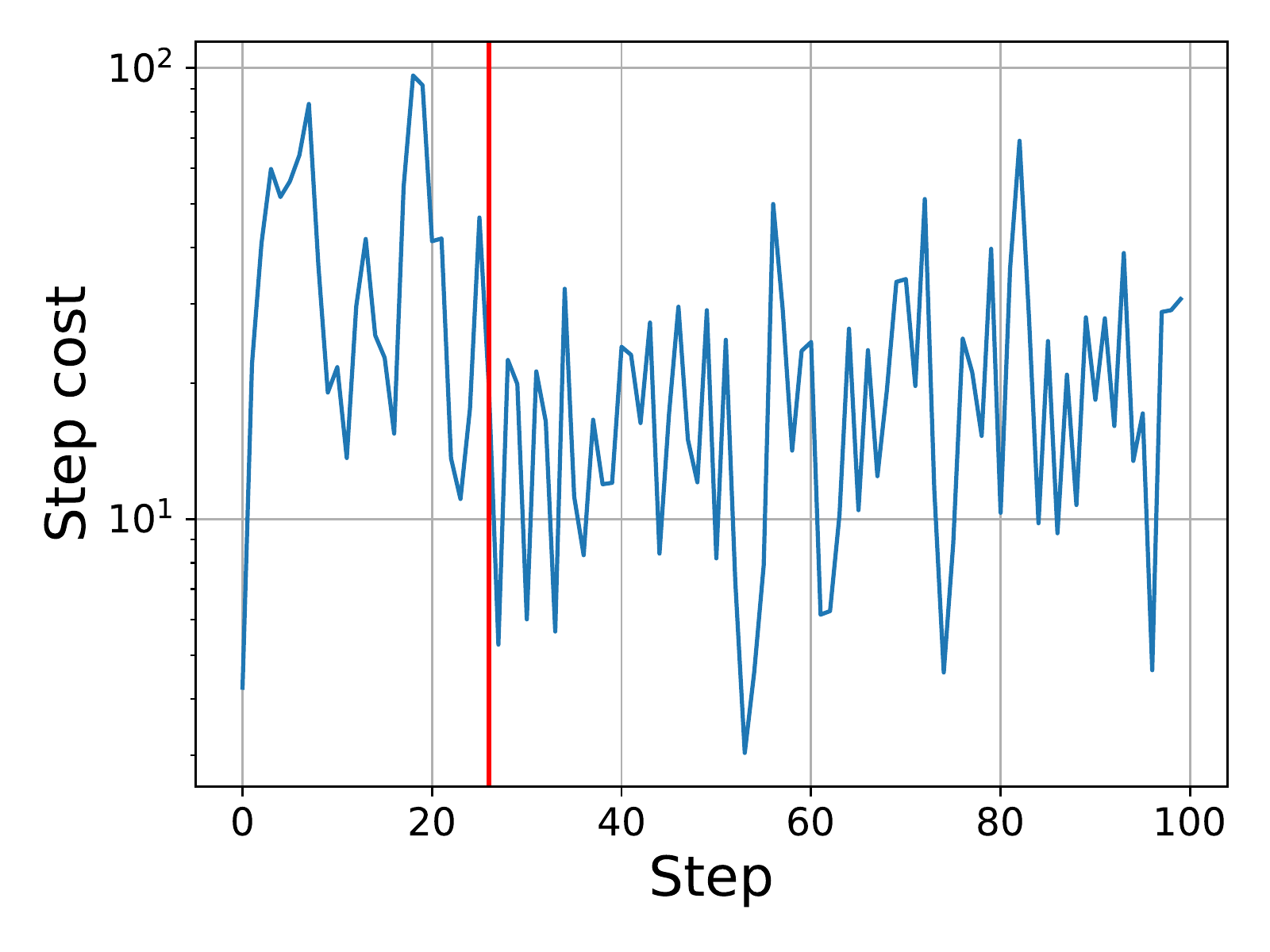}
    }   
    \end{subfigure}
    \begin{subfigure}[$K_i$ as \textsc{RelaxedSDP} controller]{
        \includegraphics[width=0.5\textwidth]{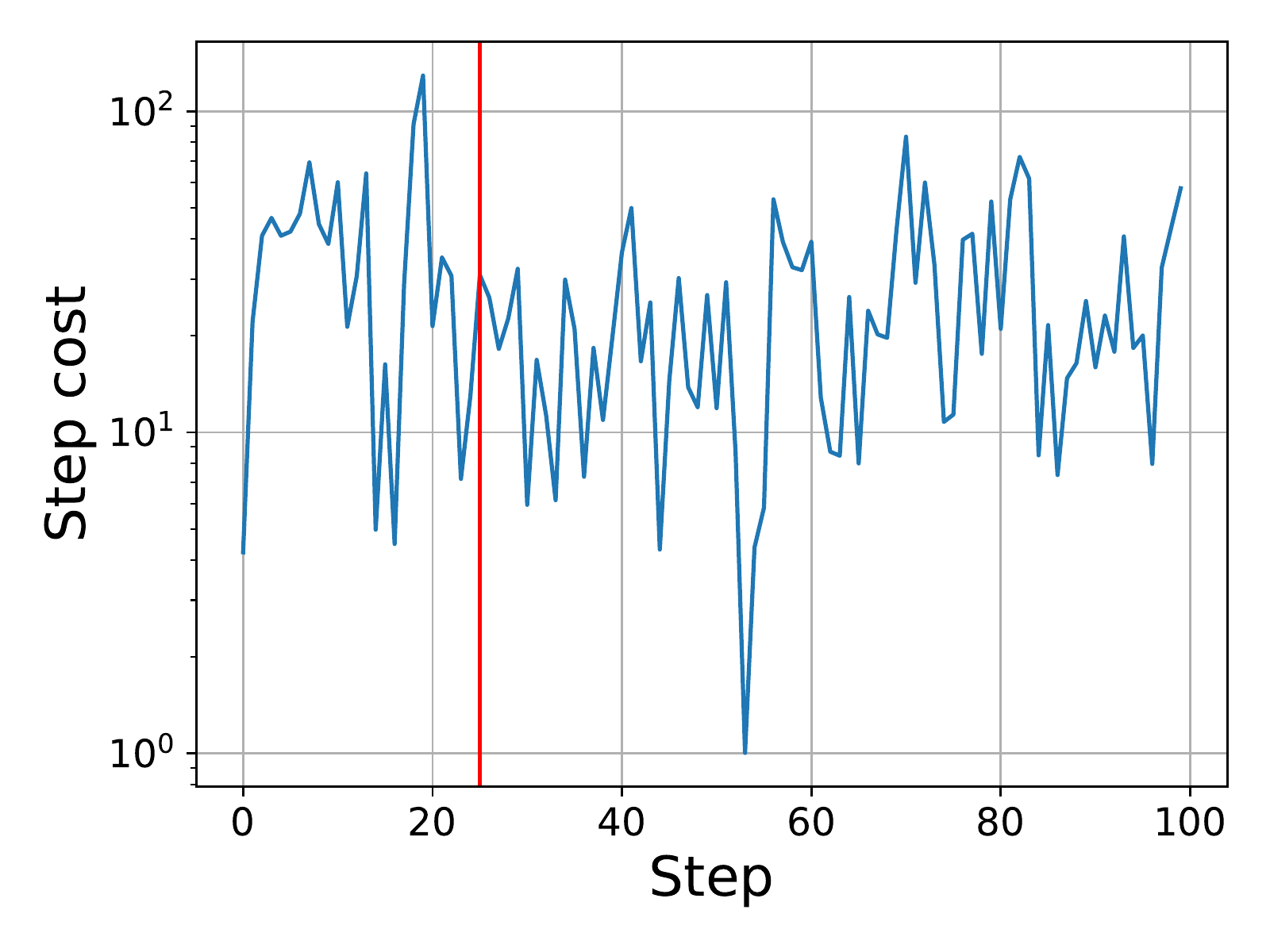}
    }   
    \end{subfigure}
    \caption{Using different controllers before we find a stabilizing controller reduces the initial blowup. When we do not choose zero mean actions we do not have a theoretical guarantee that \textsc{eXploration} will eventually terminate, however experiments show that using controller before stabilization does not harm the termination time of \textsc{eXploration}.}
    \label{fig: sample run with different controllers on system Dean et al}
\end{figure}

Next we provide a sample run on bit more explosive system:
\begin{align}
\label{equation:LargerAndMoreExplosiveSystem}
    A_* = 
    \begin{pmatrix}
        1.5 & 1.0 & 0.4 & 2.3 \\
        0.0 & 1.3 & 1.3 & 1.1 \\
        0.0 & 0.0 & 1.0 & 0.7 \\
        0.0 & 0.0 & 0.0 & 0.8
    \end{pmatrix}
    ,\quad
    B_* = 
    \begin{pmatrix}
        0.6 & 0.7 & 0.3 \\
        0.8 & 1.1 & 1.1 \\
        1.2 & 0.2 & 2.3 \\
        2.1 & 0.4 & 0.4 \\
    \end{pmatrix},
\end{align}
which we present on \Cref{fig: sample run with different controllers on explosive system}. We see that the time it takes to find a stabilizing controller is very short. Usually of the order of system dimension $d_x + d_u$ (this is the time it takes for Grammian matrix $V_i$ to be invertible). We also observe that the initial blow-up is considerably smaller if we use any of the proposed data-dependent controllers $K_i$.

\begin{figure}[ht]
    \begin{subfigure}[$K_i = 0$]{
        \includegraphics[width=0.5\textwidth]{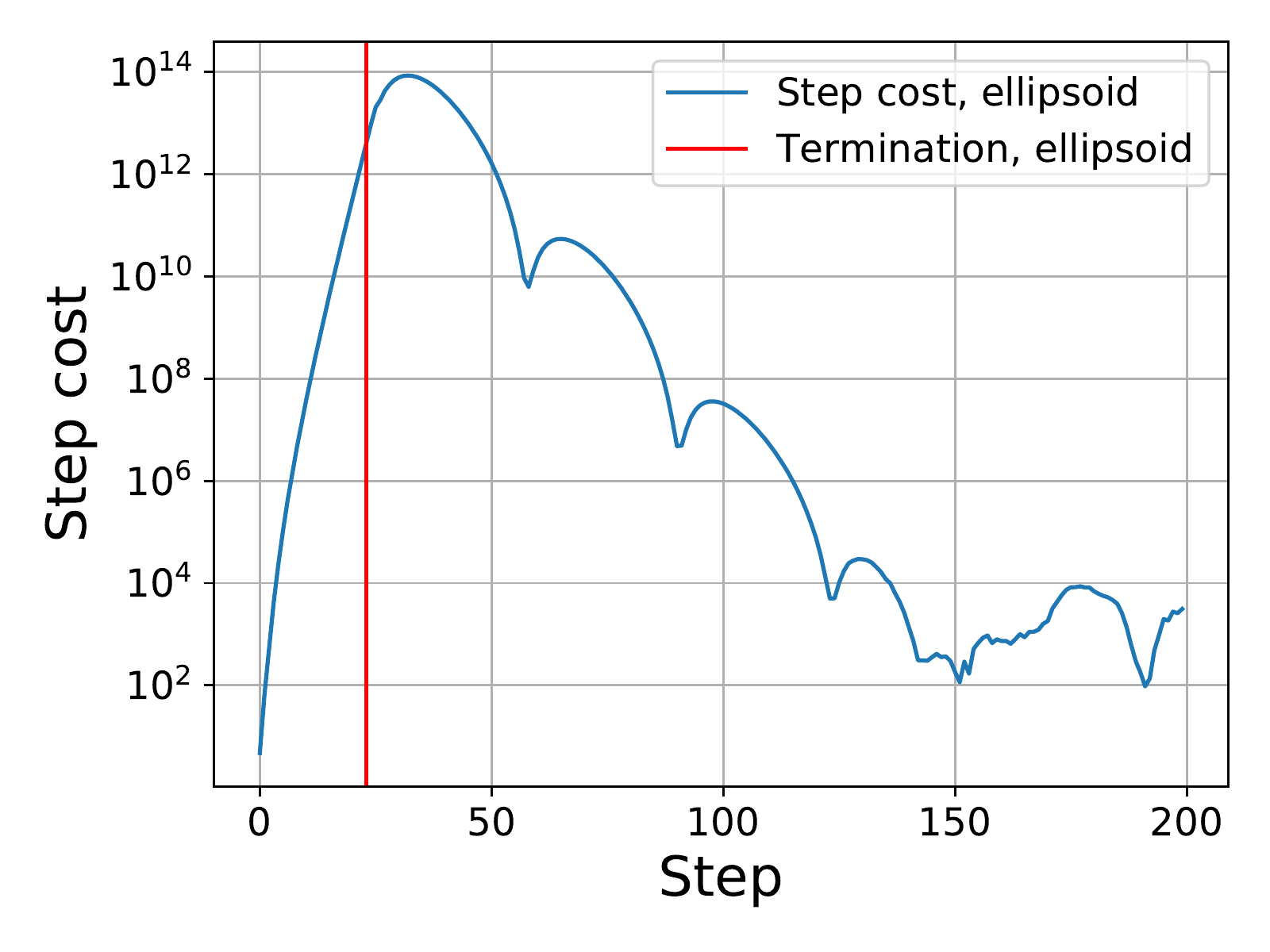}
    }
    \end{subfigure}
    \begin{subfigure}[$K_i$ as \textsc{CE} controller]{
        \includegraphics[width=0.5\textwidth]{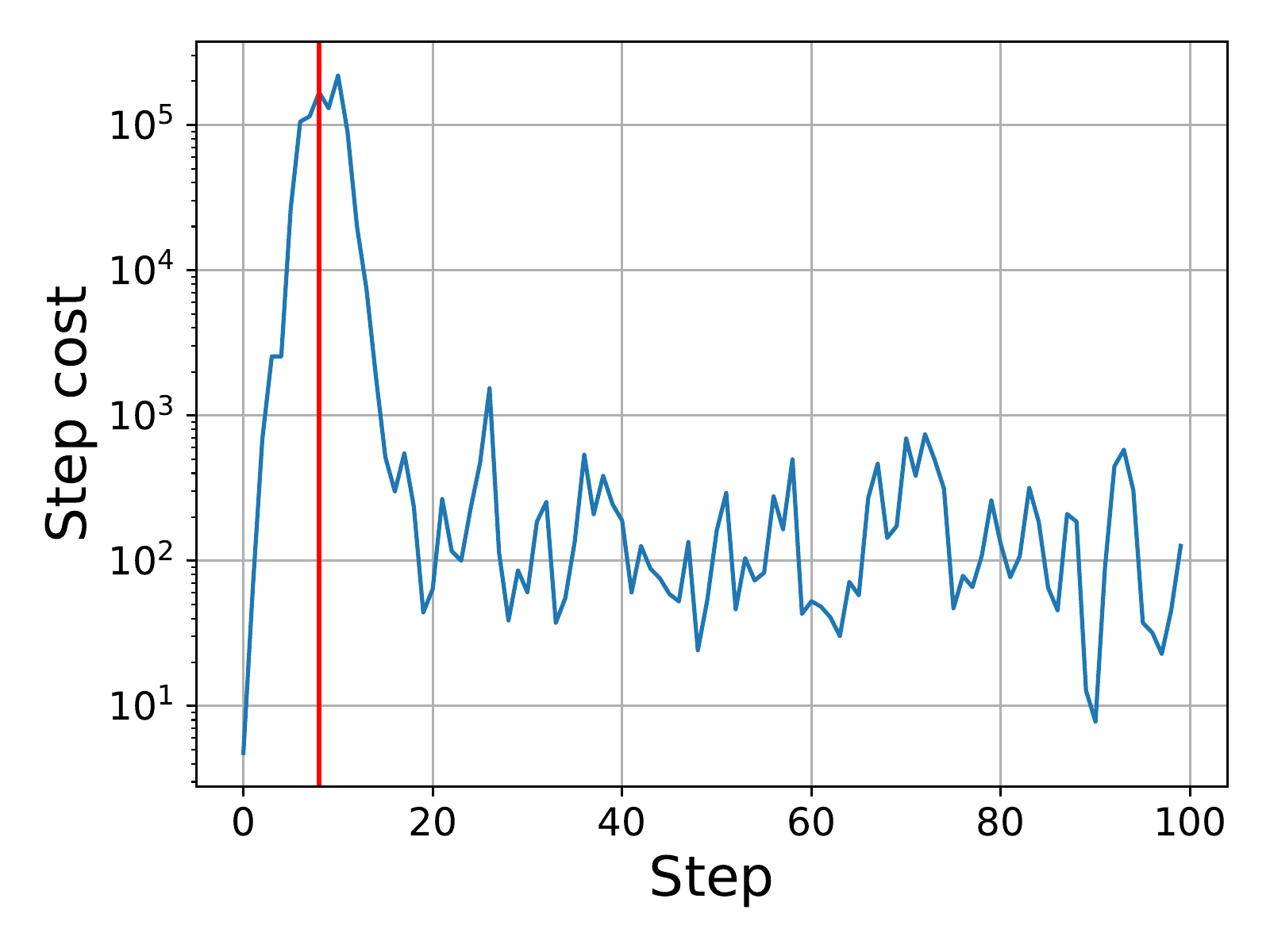}
    }  
    \end{subfigure}
    \begin{subfigure}[$K_i$ as \textsc{MinMax} controller]{
        \includegraphics[width=0.5\textwidth]{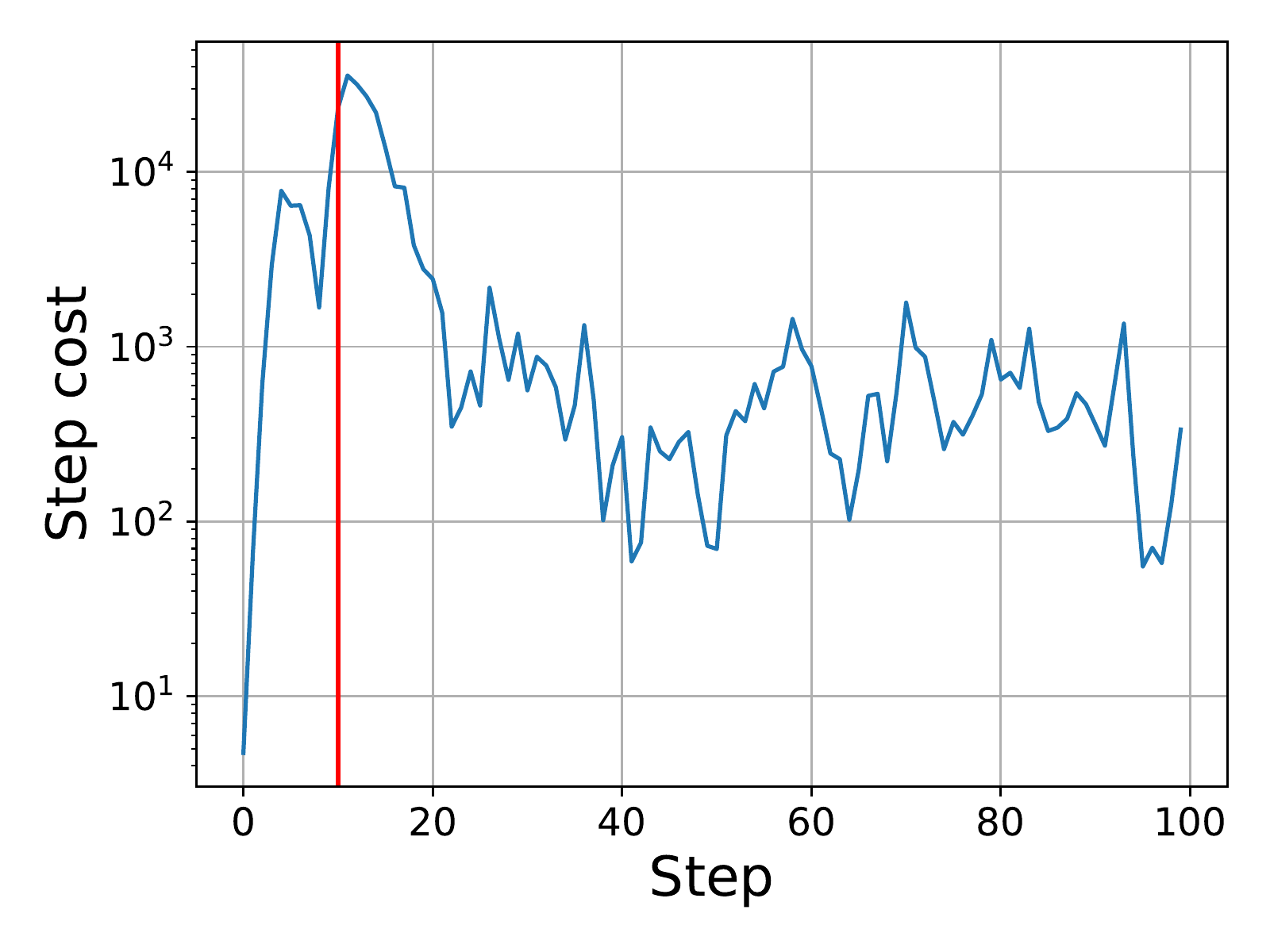}
    }   
    \end{subfigure}
    \begin{subfigure}[$K_i$ as \textsc{RelaxedSDP} controller]{
        \includegraphics[width=0.5\textwidth]{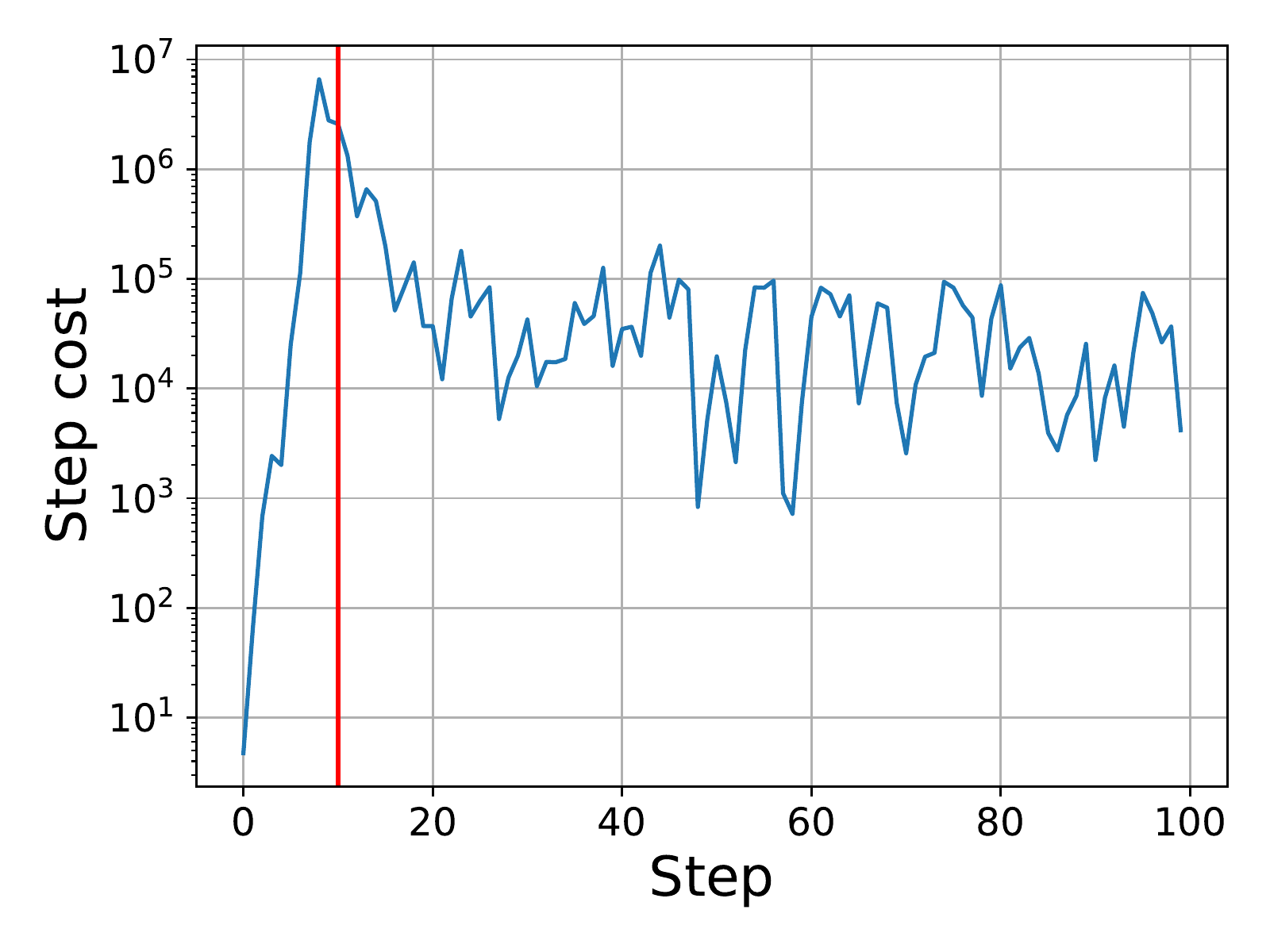}
    }   
    \end{subfigure}
    \caption{On more explosive systems the difference between the case when we use data dependent controllers $K_i$ or inject zero-mean Gaussian actions is more significant.}
    \label{fig: sample run with different controllers on explosive system}
\end{figure}

It is well known that in the classical LQR setting the optimal controller is robust to some extent i.e. it stabilizes also some region around the true estimates. In the next experiment we compare its robustness to the robust controllers described in \Cref{subsection: Robust formulation from System Level Synthesis} and \Cref{subsection: Robust formulation from Semi-Definite Program}. We sample uniformly at random 5 systems $(A, B) \in [-3,3]\times[-3,3]$. Then we analyze the performance of the optimal and robust controller on systems within some ball around them. For every radius we plot the largest infinite horizon cost of any system inside the ball if we played the proposed controllers. As we can see on \Cref{fig: comparison of robust and optimal controller}, the robust controller always stabilizes larger region, on System 2 on \Cref{fig: comparison of robust and optimal controller} it stabilizes region with almost twice as large radius as the optimal controller.

\begin{figure}[H]
        \includegraphics[width=0.5\textwidth]{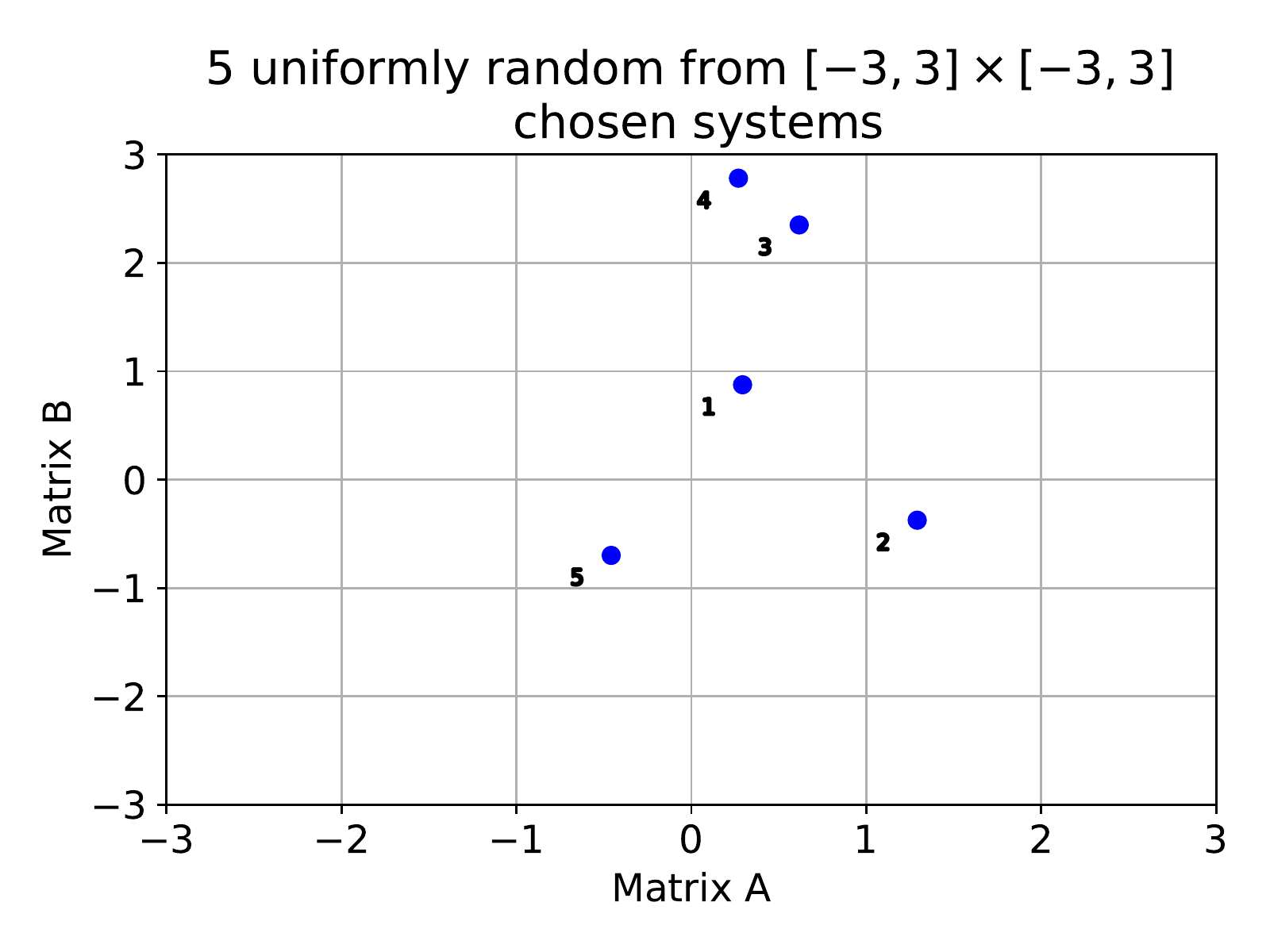}
        \includegraphics[width=0.5\textwidth]{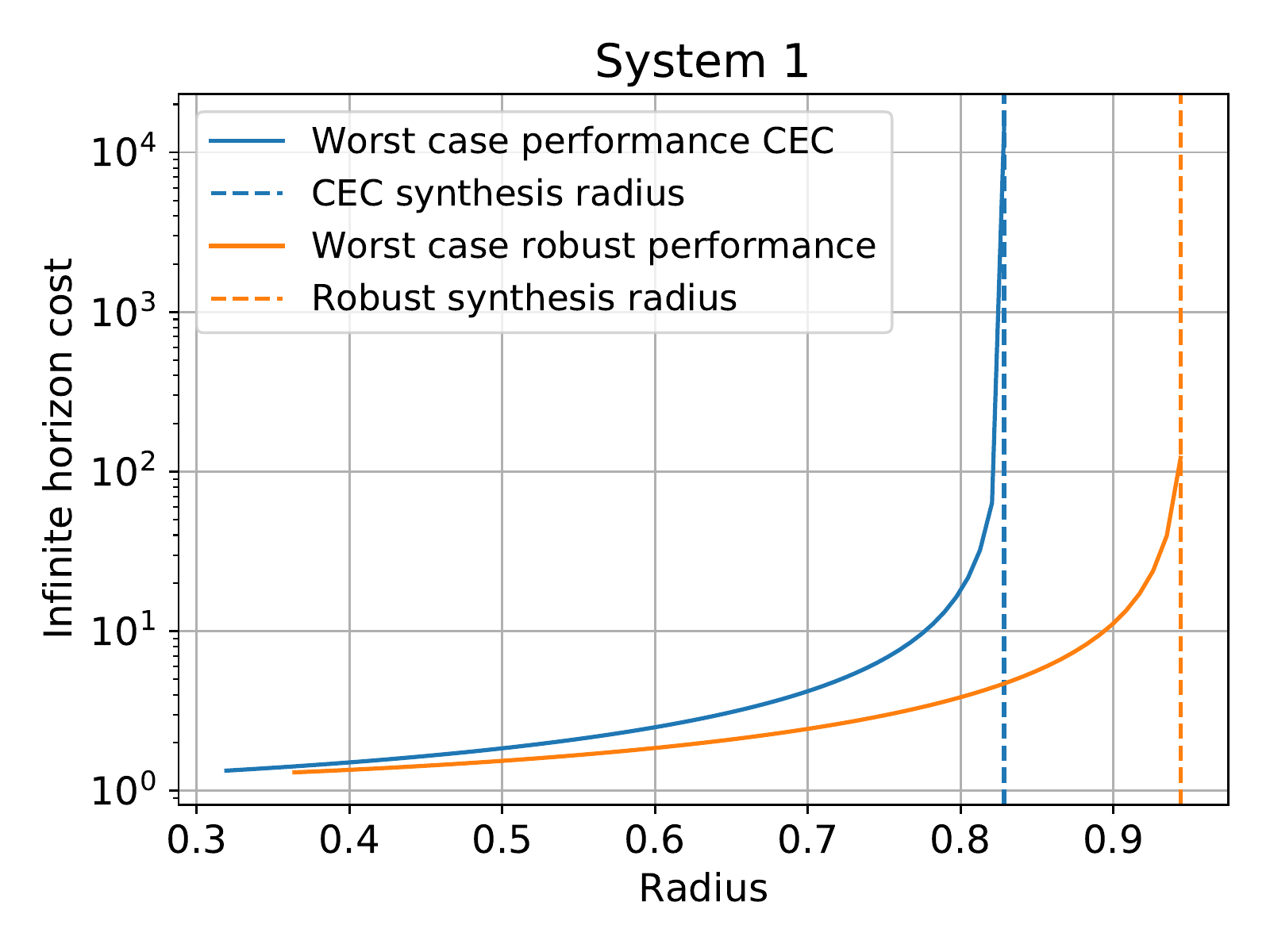}
        \includegraphics[width=0.5\textwidth]{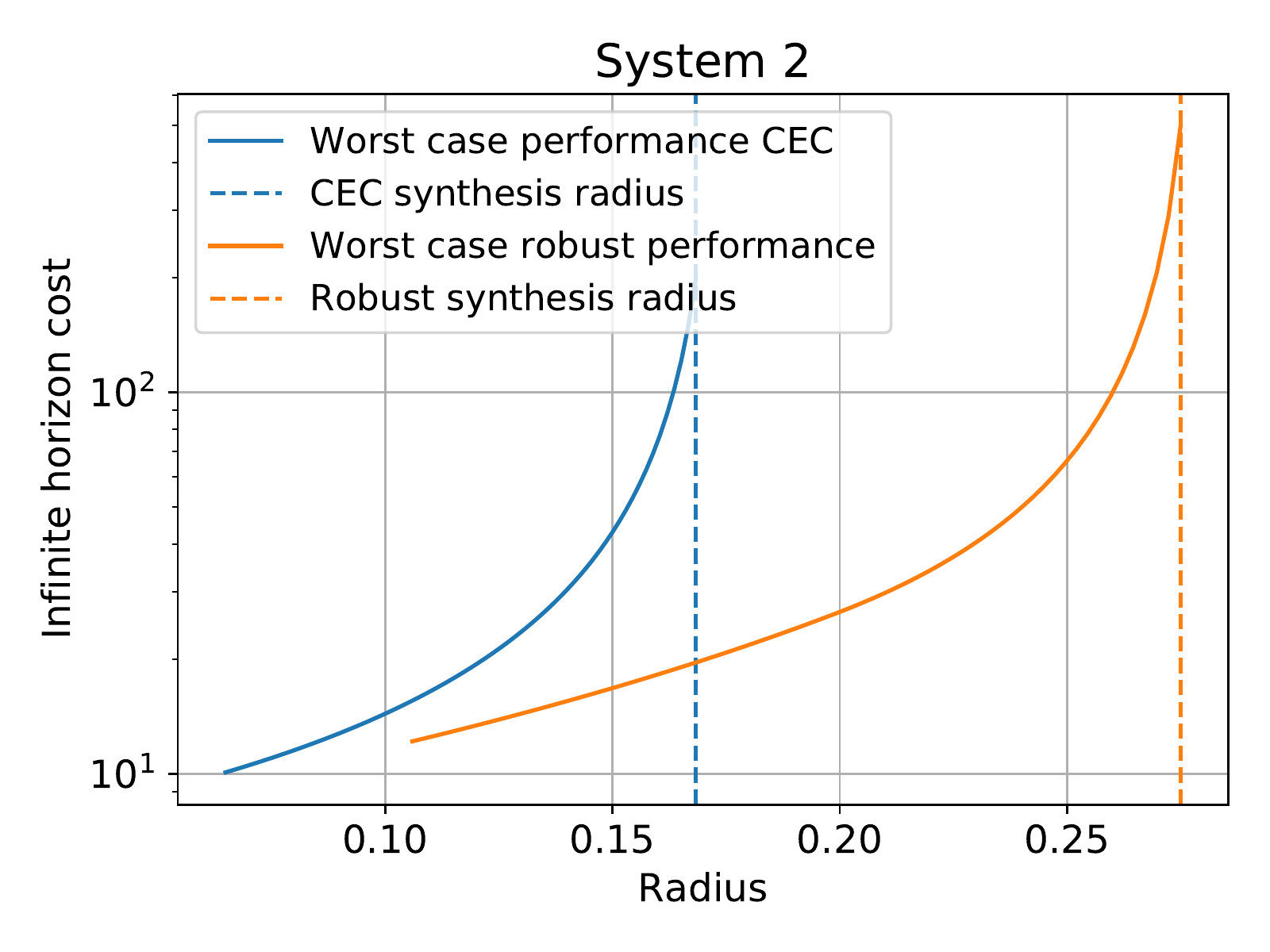}
        \includegraphics[width=0.5\textwidth]{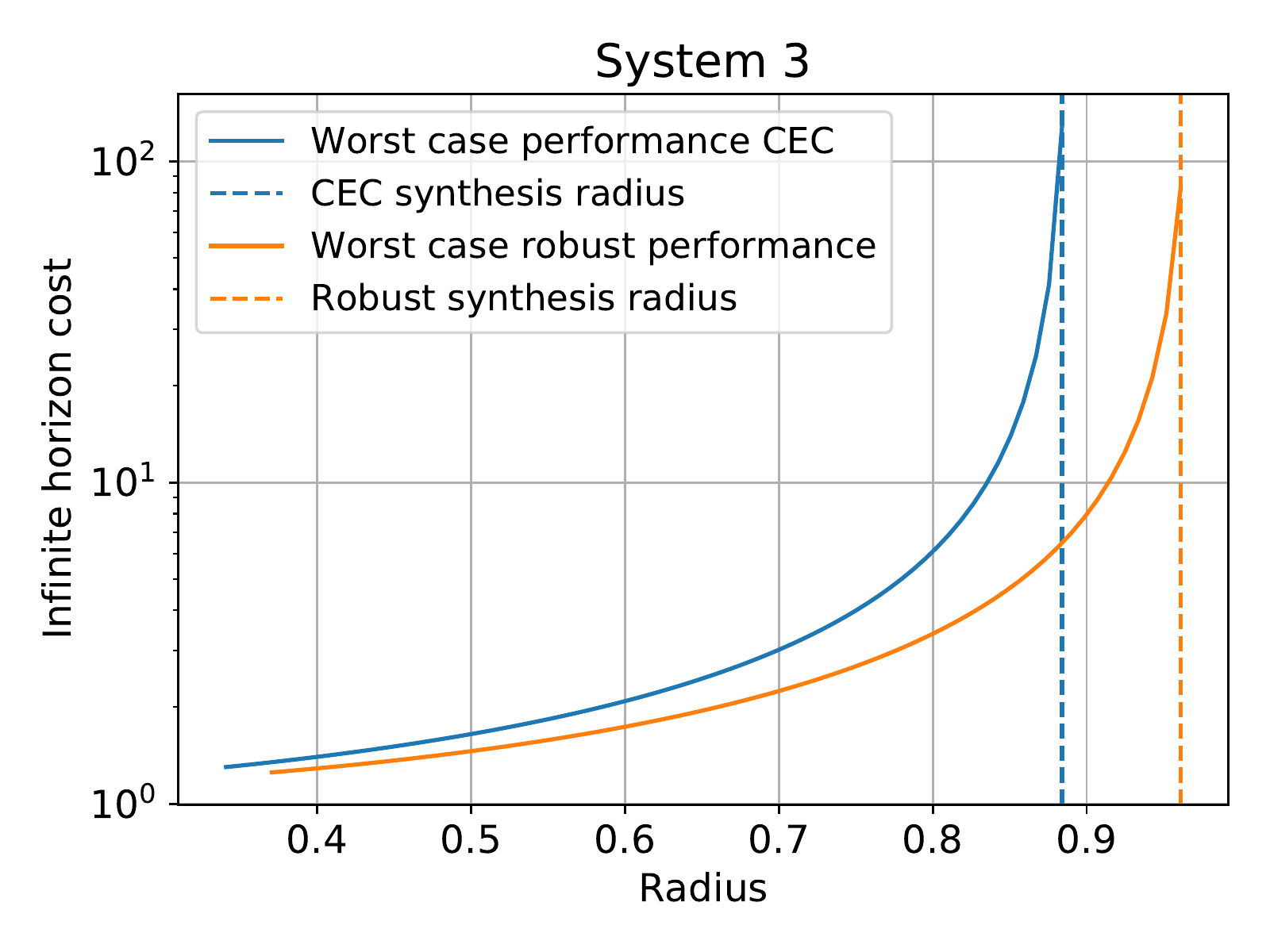}
        \includegraphics[width=0.5\textwidth]{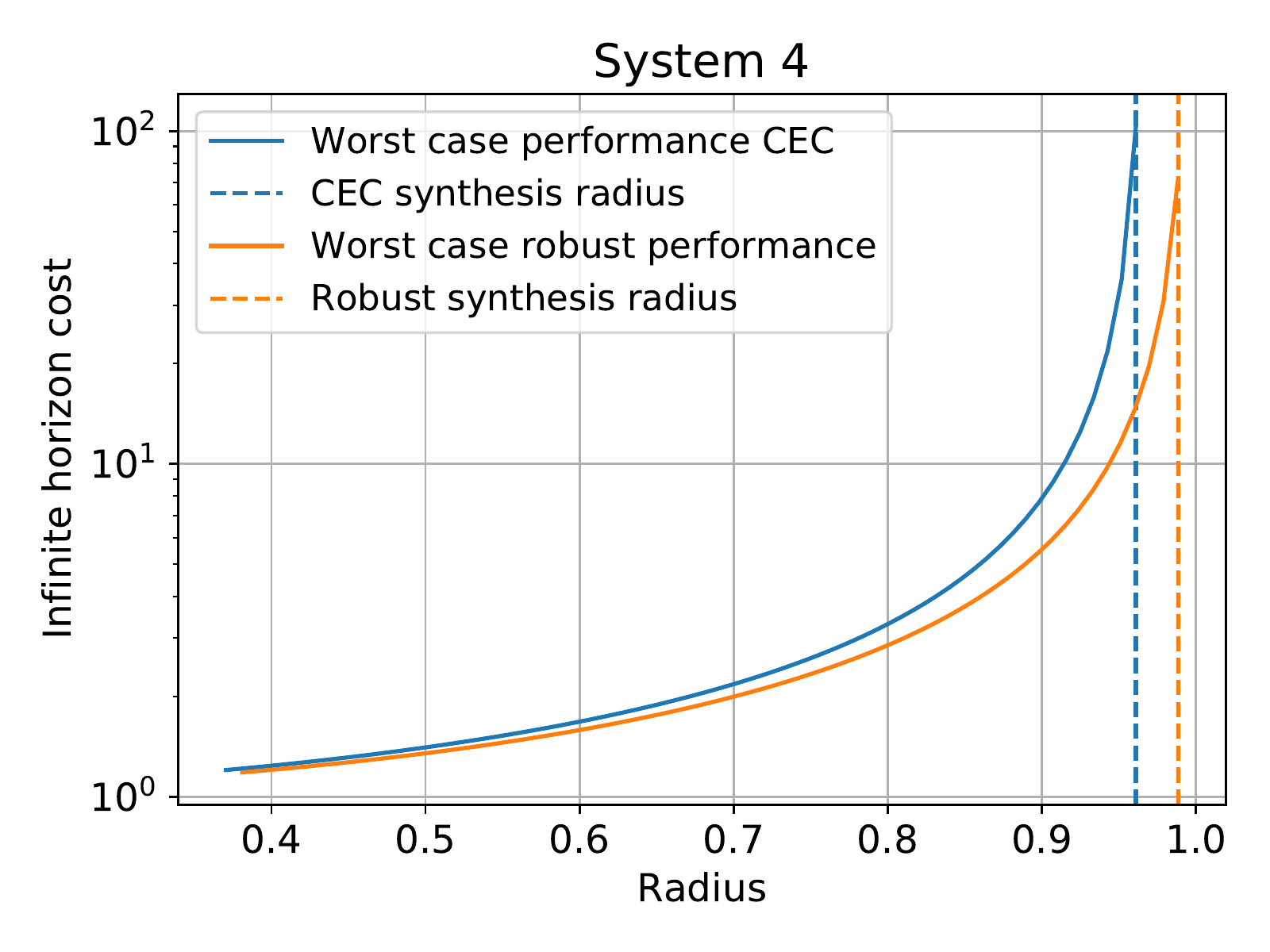}
        \includegraphics[width=0.5\textwidth]{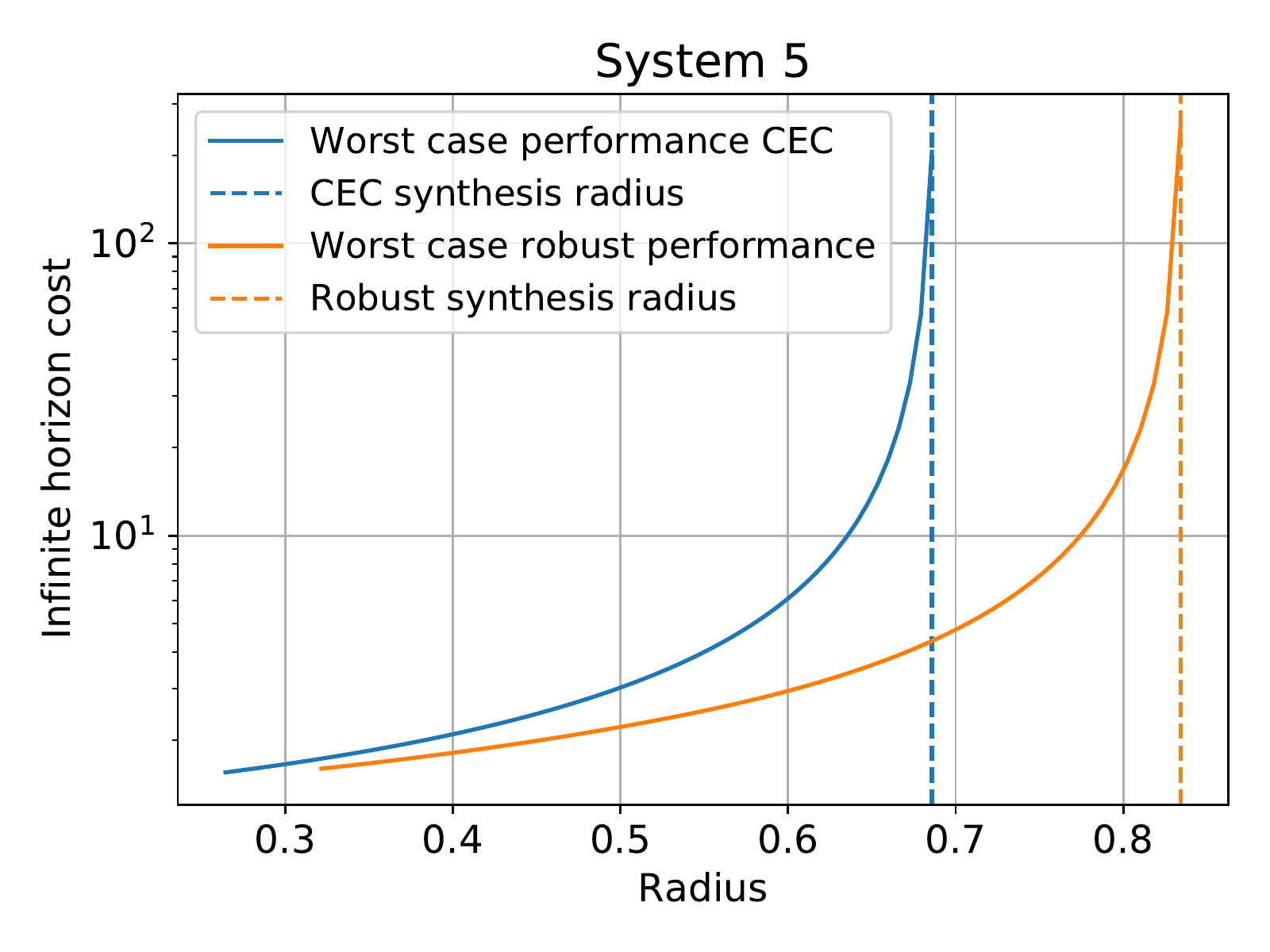}
    \caption{We sampled 5 systems and then computed the worse case performance on the systems within some radius around them. As we can see the robust controller can stabilize all systems within considerably larger radius around true system compared to CEC.}
    \label{fig: comparison of robust and optimal controller}
\end{figure}

\newpage
\section{Bounds discussion}
\label{section: Bounds failure}

We empirically observed that setting $\lambda = \frac{2\sigma_w^2\Gamma(\frac{n}{2}+1)}{C^2}$ yields that credibility regions, with $\delta$ probability of failure, contain systems with $\norm{(A_*, B_*)}_F \le C$ with empirical probability at least $1-\delta$. Next we will motivate the decision behind such $\lambda$ selection.

If the prior belief about the system is uniformly distributed over the set $I = \{(A, B)|\norm{(A~B)}_F \le C\}$, then in order for prior Gaussian probability density function to be smaller or equal on $I$ than prior uniform probability density function we need to choose $\lambda$ such that:
\begin{align*}
    \frac{1}{(2\pi)^{\frac{n}{2}}\left(\frac{\sigma_w^2}{\lambda}\right)^{\frac{n}{2}}} \le \frac{\Gamma(\frac{n}{2}+1)}{\pi^{\frac{n}{2}}C^n},
\end{align*}
which is satisfied if we select $\lambda = \frac{2\sigma_w^2\Gamma(\frac{n}{2}+1)}{C^2}$. As can be seen from \Cref{figure: lambda selection} this selection might not be the optimal one, however it seems that it at least captures the order of behavior $\lambda \approx 1/C^2$.

The ultimate goal would be to obtain data dependent consistent (for regular systems) estimation error upper bounds in the frequentist setting. If we take as the estimators $\widehat{A}, \widehat{B}$ the RLS with regularizing parameter $\lambda$ one can show that:
\begin{align*}
    \left((\widehat{A}~\widehat{B}) - (A_*~B_*)\right)^\top &= \left(V_i + \lambda I\right)^{-1}S_i -  \lambda\left(V_i + \lambda I\right)^{-1}(A_*~B_*)^\top,
\end{align*}
where $V_i = \sum_{j=0}^{i-1}z_jz_j^\top$ and $S_i = \sum_{j=0}^{i-1}z_jw_{j+1}^\top$, with $z_j = (x_j^\top u_j^\top)^\top$. \citet{Sarkar2018Identification} showed that for regular systems $\norm{\left(V_i + \lambda I\right)^{-1}S_i -  \lambda\left(V_i + \lambda I\right)^{-1}(A_*~B_*)^\top}_2 \stackrel{i \to \infty}{\to} 0$. However the term $S_i$ is not observed.
One could try and use the theory of Self-Normalizing Martingales and show that w.p. $1-\delta$:
 \begin{align}
 \label{self-normalizing martingales}
    \norm{(V_i + \lambda I)^{-\frac{1}{2}}S_i}^2_2 \le 8\sigma_w^2\log \left(\frac{\det(V_i + \lambda 
    I)^{\frac{1}{2}}}{\det(\lambda I )^{\frac{1}{2}}}\frac{5^d}{\delta}\right),    
 \end{align}
however, since the norm of states $x_i$ can grow exponentially, one could show that the right hand side of \Cref{self-normalizing martingales} can grow linearly and hence the upper bound on estimation error could be inconsistent. \citet{Sarkar2018Identification} used the theory of Self-Normalizing Martingales to show that OLS is consistent, however instead of $\lambda I$ in \Cref{self-normalizing martingales} they inserted $V_{dn}$ with $V_{dn} \preceq V_i$, however $V_{dn}$ significantly depends on the system $(A, B)$ and is to the best of or knowledge not known how to obtain it in the data-dependent setting.
How to get consistent data dependent upper bounds for the error term $\norm{\left(V_i + \lambda I\right)^{-1}S_i -  \lambda\left(V_i + \lambda I\right)^{-1}(A_*~B_*)^\top}_2$ is to the best of our knowledge also not known. However as we argued in \Cref{section: Experiments}, we observed empirically that with the right prior selection of $\lambda$, based on the knowledge of $C$ with $\norm{(A_*~B_*)}_F \le C$, we can use Bayesian credibility regions $\Theta$.

\begin{figure}[H]
    \centering
    \begin{subfigure}[$d_x =2, d_u = 1$]{
        \includegraphics[width=0.30\textwidth]{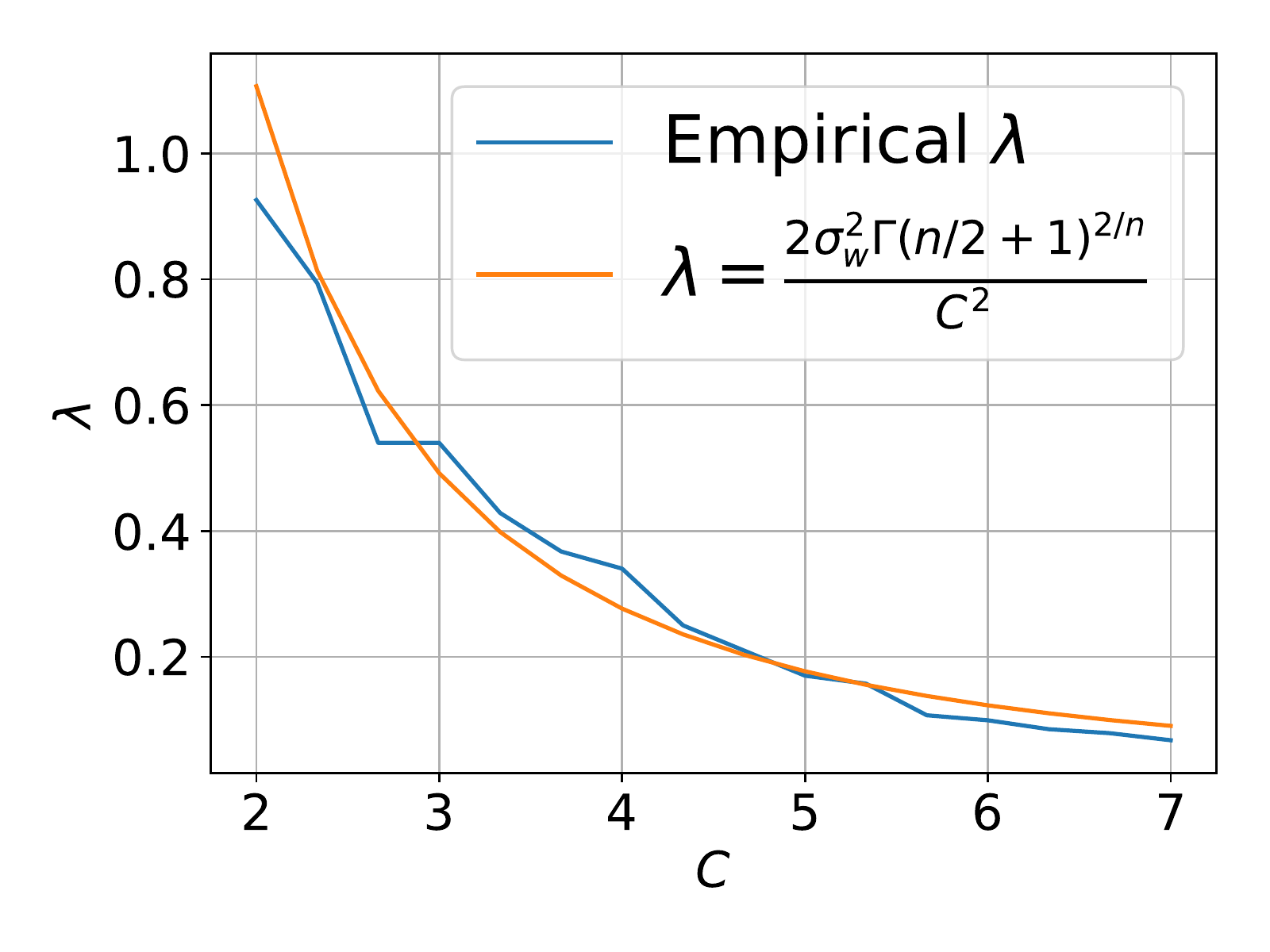}
        \label{figure: d_x = 2, d_u = 1}
    }
    \end{subfigure}
    \begin{subfigure}[$d_x=d_u=2$]{
        \includegraphics[width=0.30\textwidth]{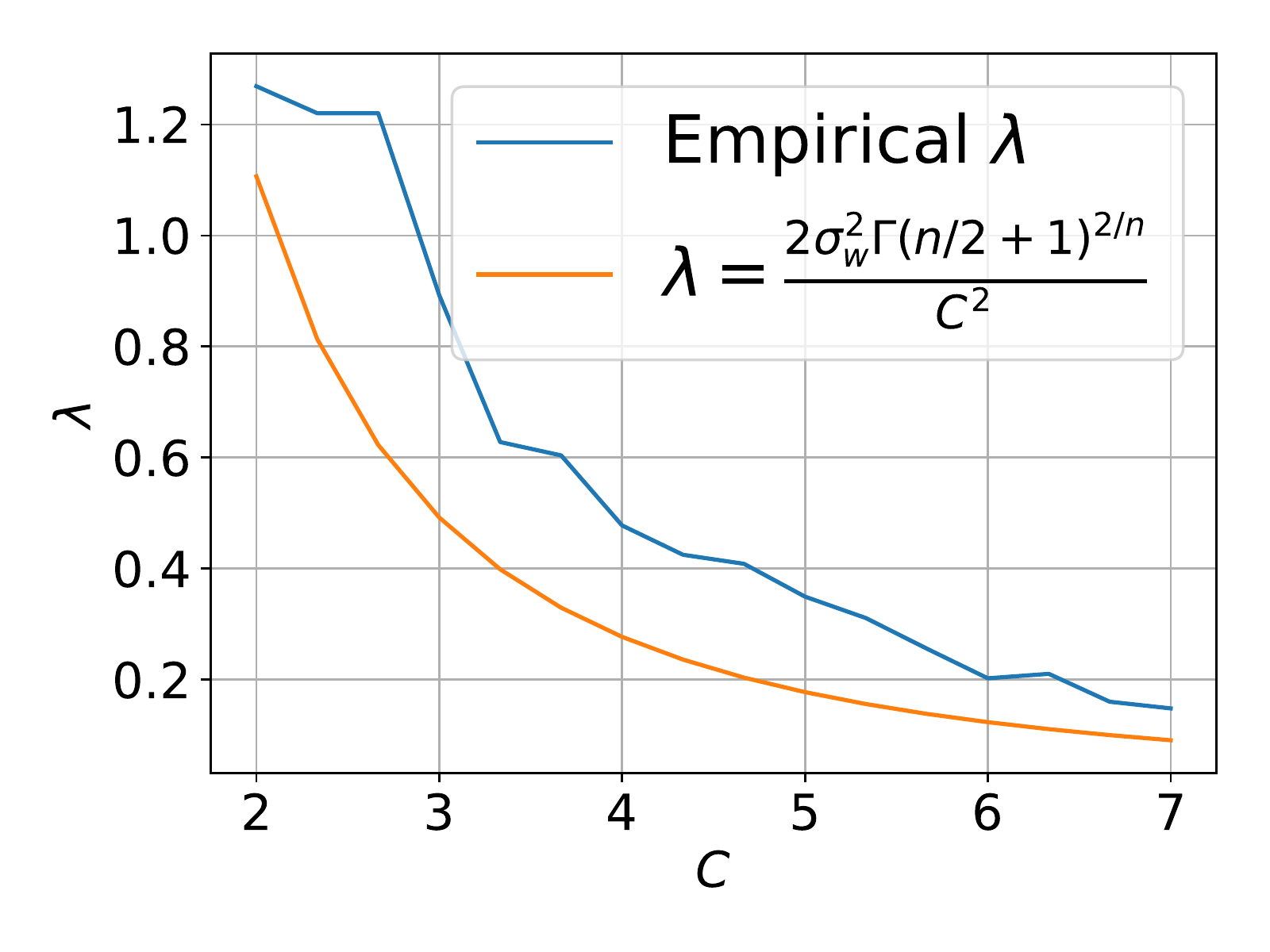}
        \label{figure: d_x, d_u = 2}
    }
    \end{subfigure}
    \begin{subfigure}[$d_x=d_u=3$]{
        \includegraphics[width=0.30\textwidth]{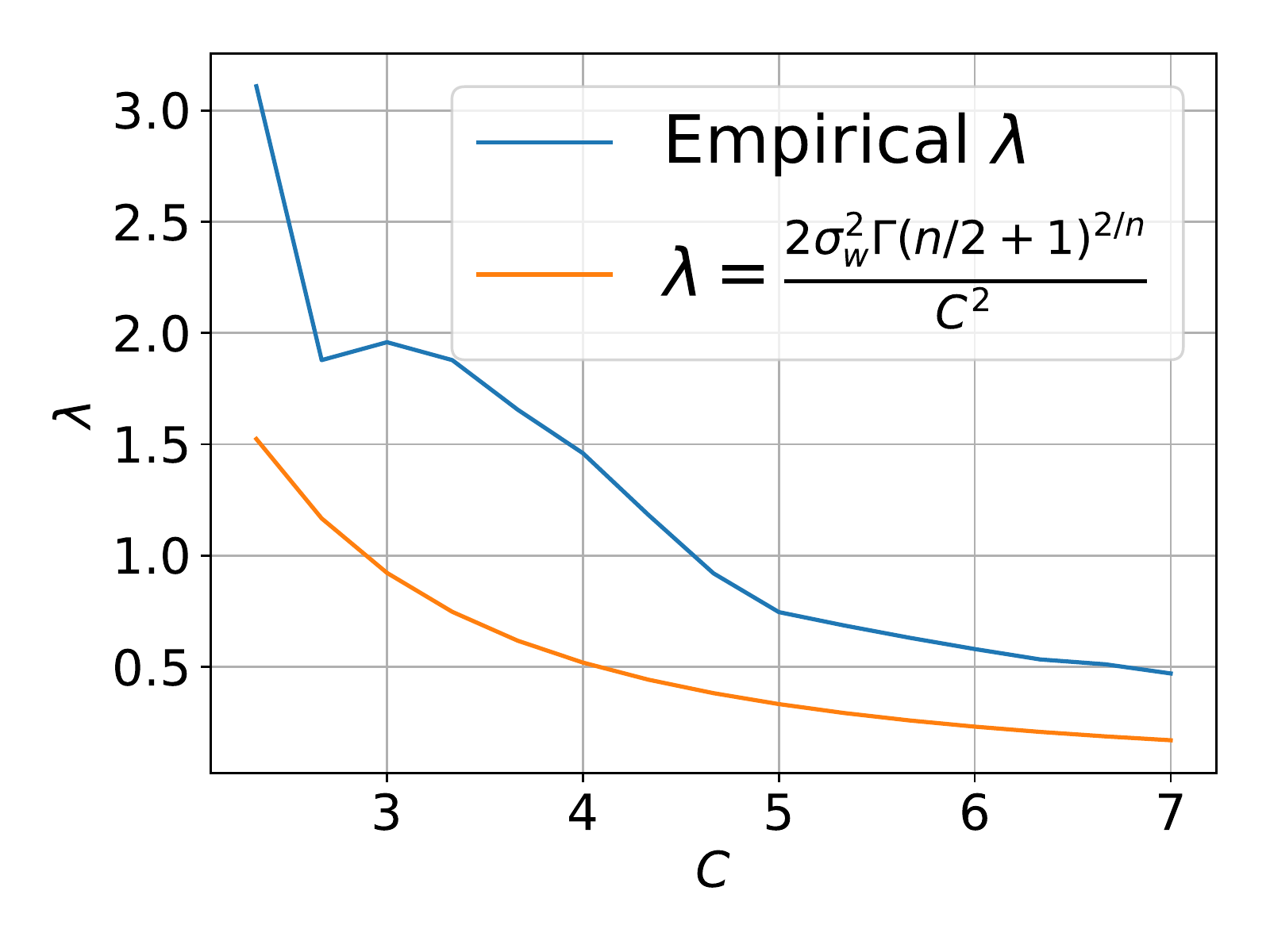}
        \label{figure: d_x=d_u=3}
    }
    \end{subfigure} 
    \caption{If we select $\lambda$ as described above, it empirically turns out to be small enough, such that region $\Theta$ contains $A_*, B_*$.}
    \label{figure: lambda selection}
\end{figure}

In the experiment presented in right most picture \Cref{figure: moving credibility regions} and \Cref{figure: lambda selection} we first selected different $C$-s for which we computed the "optimal" $\lambda$ ($\lambda$ presented on the figures as Empirical $\lambda$.) How we computed "optimal" $\lambda$? Consider fixed $C$ and $\lambda$. We sampled $200$ systems $(A ~B)$ uniformly at random from the set $\{(A~B)|\norm{(A~B)}_F = C\}$ and evolved each system for 10 steps. Then we computed the share of the steps when $(A~B)$ were inside region $\Theta$. The "optimal" $\lambda$ is the largest $\lambda$ where the share is larger than $1-\delta$.
\end{document}